%% file: main.tex
\documentclass[journal]{IEEEtran}

\usepackage{wrapfig}
\usepackage[T1]{fontenc}
\usepackage[bookmarks=false]{hyperref}
\usepackage[cmex10]{amsmath}
\usepackage{algorithm}
\usepackage[noend]{algpseudocode}
\usepackage{pgfplots}
\makeatletter
\def\BState{\State\hskip-\ALG@thistlm}
\makeatother
\usepackage{algpseudocode}
\usepackage{cite}
\usepackage{algorithm}

\usepackage{graphicx,subcaption}
\usepackage{pstricks}
\usepackage{color}
\usepackage{amssymb}
\usepackage{amsthm}
\usepackage{pgfplots}
\usepackage{tikz}
\usepackage{relsize}
\usepackage{bbm}
\usepackage{bm}
\usepackage[font={small}]{caption}
\usepackage{comment,color,soul}
\usepackage[normalem]{ulem}
\usepackage{pst-plot}
\usepackage{tkz-euclide}
\usepackage{pgfplots}
\usepackage{mathtools}

\usepackage[noend]{algpseudocode}
\usepackage{graphicx}
\newsavebox\CBox
\def\textBF#1{\sbox\CBox{#1}\resizebox{\wd\CBox}{\ht\CBox}{\textbf{#1}}}

\makeatletter
\def\algbackskip{\hskip-\ALG@thistlm}
\makeatother

\pgfplotsset{compat=1.11}
\usetikzlibrary{calc,positioning,arrows.meta,quotes}
\usepackage{lipsum}

\makeatletter
\newcommand{\doublewidetilde}[1]{{%
  \mathpalette\double@widetilde{#1}%
}}
\newcommand{\double@widetilde}[2]{%
  \sbox\z@{$\m@th#1\widetilde{#2}$}%
  \ht\z@=.9\ht\z@
  \widetilde{\box\z@}%
}
\makeatother

\makeatletter
\def\thm@space@setup{\thm@preskip=2pt
        \thm@postskip=2pt \itshape}
\makeatother

\newtheoremstyle{newstyle}
{} 
{} 
{\mdseries} 
{} 
{\bfseries} 
{.} 
{ } 
{} 

\theoremstyle{newstyle}

\newtheorem{theorem}{Theorem}
\newtheorem{lemma}{Lemma}

\theoremstyle{definition}

\newtheorem*{example*}{Example}
\newtheorem{definition}{Definition}

\theoremstyle{remark}
\newtheorem{remark}{Remark}

\IEEEoverridecommandlockouts
\newcommand{\Expc}{\mathbb{E}}
\newcommand{\Prob}{\Pr}
\newcommand{\bA}{\mathbf{A}}
\newcommand{\bx}{\mathbf{x}}

\newcommand*\Bell{\ensuremath{\boldsymbol\ell}}

\begin{document}

\title{Coded Computation over Heterogeneous Clusters}
\author{Amirhossein Reisizadeh, Saurav Prakash, Ramtin Pedarsani, Amir Salman Avestimehr
    \thanks{A. Reisizadeh and R. Pedarsani are with the Department
    of Electrical and Computer Engineering, University of California, Santa Barbara, Santa Barbara, CA 93106 USA (e-mail: reisizadeh@ucsb.edu; ramtin@ece.ucsb.edu).}
    \thanks{S. Prakash and A. S. Avestimehr are with the the Department
    of Electrical and Computer Engineering, University of Southern California, Los Angeles, CA 90089 USA
    (e-mail: sauravpr@usc.edu; avestimehr@ee.usc.edu).}
    \thanks{A part of this work was presented in IEEE International Symposium on Information Theory, 2017 \cite{reisizadeh2017coded}.}
   }



\maketitle

\begin{abstract} 
\input{0-abstract.tex}

\end{abstract}

\begin{IEEEkeywords}
Coded computation, distributed computing, heterogeneous clusters.
\end{IEEEkeywords}

\IEEEpeerreviewmaketitle

\section{Introduction}
\input{1-intro.tex}

\section{Problem Formulation and Main Results}
\input{2-prob_form_main_res.tex}
\section{The Proposed HCMM Scheme and Proofs of Theorems \ref{theorem1} and \ref{theorem2}}
\label{HCMM}
\input{3-HCMM_scheme.tex}
\section{Generalization to the Shifted Weibull Model and Proofs of Theorems \ref{theorem3} and \ref{theorem4}}\label{sec:weibull}
\input{4-gen_weibull.tex}

\section{Numerical Studies and Experiments using Amazon EC2 Machines}
\label{simulation}
\input{5-experiments.tex}

\section{Generalization to Computing Scenarios under Budget Constraints}

\input{6-gen_budget_constr}

\section{Conclusion}
\input{8-conclusion}

\section{Acknowledgment}
\input{11-ack.tex}

\bibliographystyle{ieeetr}
\bibliography{biblio}
\appendix
\input{9-appendices.tex}

\input{10-bio.tex}

\end{document}

%% file: 0-abstract.tex
In large-scale distributed computing clusters, such as Amazon EC2,  there are several types of ``system noise'' that can result in major degradation of performance: system failures, bottlenecks due to limited communication bandwidth, latency due to straggler nodes, etc. There have been recent results that demonstrate the impact of coding for efficient utilization of computation and storage redundancy to alleviate the effect of stragglers and communication bottlenecks in \emph{homogeneous} clusters. In this paper, we focus on general \emph{heterogeneous} distributed computing clusters consisting of a variety of computing machines with different capabilities. We propose a coding framework for speeding up distributed computing in heterogeneous clusters by trading redundancy for reducing the latency of computation. In particular, we propose  \emph{Heterogeneous Coded Matrix Multiplication (HCMM)} algorithm for performing distributed matrix multiplication over heterogeneous clusters that is provably asymptotically optimal for a broad class of processing time distributions. Moreover, we show that HCMM is unboundedly faster than \textit{any uncoded} scheme that partitions the total work load among the workers. To demonstrate how the proposed HCMM scheme can be applied in practice, we provide results from numerical studies and Amazon EC2 experiments comparing HCMM with three benchmark load allocation schemes -- Uniform Uncoded, Load-balanced Uncoded, and Uniform Coded. In particular, in our numerical studies, HCMM achieves speedups of up to $73\%$, $56\%$ and $42\%$  respectively over the three benchmark schemes mentioned above. Furthermore, we carry out experiments over Amazon EC2 clusters and demonstrate how HCMM can be combined with rateless codes with nearly linear decoding complexity. In particular, we show that HCMM combined with the Luby transform (LT) codes can significantly reduce the overall execution time. HCMM is found to be up to $61\%$, $46\%$ and $36\%$ faster than the aforementioned three benchmark schemes, respectively. Additionally, we provide a generalization to the problem of optimal load allocation in heterogeneous settings, where we take into account the monetary costs associated with distributed computing clusters. We argue that HCMM is asymptotically optimal for budget-constrained scenarios as well. In particular, we characterize the minimum possible expected cost associated with a computation task over a given cluster of machines. Furthermore, we develop a heuristic algorithm for (HCMM) load allocation for the  distributed implementation of budget-limited computation tasks.

%% file: 1-intro.tex
\IEEEPARstart{G}{eneral}  distributed computing frameworks, such as MapReduce~\cite{dean2008mapreduce} and Spark~\cite{zaharia2010spark}, along with the availability of large-scale commodity servers, such as Amazon EC2, have made it possible to carry out large-scale data analytics at the production level. These ``virtualized data centers'' enjoy an abundance of storage space and computing power, and are cheaper to rent by the hour than maintaining dedicated data centers round the year. However, these systems suffer from various forms of ``system noise'' which reduce their efficiency: system failures, limited communication bandwidth, straggler nodes, etc. 

The current state-of-the-art approaches to mitigate the impact of system noise in cloud computing environments involve creation of some form of ``computation redundancy''.  For example, \emph{replicating} the straggling task on another available node is a common approach to deal with stragglers  \cite{zaharia2008improving}, while partial  data replication is also used to reduce the communication load in distributed computing \cite{pu2015low}. However, there have been recent results demonstrating that \emph{coding}  can play a transformational role for creating and exploiting computation redundancy to effectively alleviate the impact of system noise. In particular, there have been two coding concepts proposed to deal with the communication and straggler bottlenecks in distributed computing.

 The first coding concept introduced in~\cite{LMA_all,li2017fundamental,li2017coded} enables an inverse-linear tradeoff between computation load and communication load in distributed computing. This result  implies that increasing the computation load by a factor of $r$ (i.e. evaluating each computation at $r$ carefully chosen nodes) can create novel coding opportunities  that reduce the required communication load for computing by the same factor $r$. Hence, these codes can be utilized to pool the underutilized computing resources at network edge to slash the communication load of Fog computing \cite{bonomi2012fog}. Other related works tackling the communication bottleneck in distributed computation include \cite{lee2015speeding,song2017pliable,kiamari2017heterogeneous,attia2016information,ezzeldin2017communication}.

In the second coding concept introduced in~\cite{lee2015speeding}, an inverse-linear tradeoff between computation load  and computation latency (i.e. the overall job response time) is established for distributed matrix multiplication in homogeneous computing environments. More specifically, this approach utilizes coding to effectively inject redundant computations to alleviate the effects of stragglers and speed up the computations. Hence, by utilizing more computation resources, this can significantly speed up distributed computing applications. A number of related works have been proposed recently to mitigate stragglers in distributed computation. In \cite{dutta2016short}, the authors  propose the use of redundant short dot products to speed up distributed computation of linear transforms. The work in  \cite{tandon2016gradient} proposes coding schemes for mitigating stragglers in distributed batch gradient computation. Coding schemes for high-dimensional matrix-matrix multiplication have been developed in \cite{yu2017polynomial,fahim2017optimal,lee2017high,wang2018fundamental,yu2018straggler}. Techniques for efficient straggler mitigation for matrix-vector computation in distributed wireless settings have been developed in \cite{reisizadeh2017latency}. In \cite{lee2017coded}, the potential of the multicore nature of computing machines is studied. In \cite{ferdinand2016anytime}, the authors propose an anytime approach to distributed computing, developing an approximate matrix multiplication scheme. The authors in \cite{suh2017matrix} propose a novel encoding scheme for achieving large sparsity in the encoded matrix. Work in \cite{aliasgari2017coded} develops a coding strategy for mitigating straggling decoders in cloud radio access network. Speeding up the computation of linear transformations with unreliable components is studied in \cite{yang2017computing}. Straggler mitigation through data encoding in distributed optimization is proposed in \cite{karakus2017straggler}. A coded scheme based on LT codes is proposed in \cite{severinson2017block} for multiplying a matrix by a set of vectors in a distributed computing environment. Addressing stragglers has attracted a lot of attention in the queuing-based frameworks for large-scale computation as well \cite{aktas2017effective,wang2015using}. These works utilize the technique of dynamically replicating the tasks in a careful manner to minimize run-time.

We extend the problem of distributed matrix multiplication in homogeneous clusters in \cite{lee2015speeding} to heterogeneous environments. As discussed in \cite{zaharia2008improving}, the computing environments in virtualized data centers are heterogeneous and algorithms based on homogeneous assumptions can result in significant performance reduction. In this paper, we focus on general \emph{heterogeneous} distributed computing clusters consisting of a variety of computing machines with different capabilities. Specifically, we propose a coding framework for speeding up distributed matrix multiplication in heterogeneous clusters with straggling servers, named \emph{Heterogeneous Coded Matrix Multiplication (HCMM)}.  Matrix multiplication is a crucial computation module in many engineering and scientific disciplines. In particular, it is a fundamental component of many popular machine learning algorithms such as logistic regression, reinforcement learning and gradient descent-based algorithms. Implementations that speed up matrix multiplication would naturally speed up the execution of a wide variety of popular algorithms. Therefore, we envision HCMM to play a fundamental role in speeding up big data analytics in virtualized data centers by leveraging the wide range of computing capabilities provided by these heterogeneous environments. 

We now describe the main ideas behind HCMM, which results in asymptotically optimal performance. In a coded implementation of distributed matrix-vector multiplication, each worker node is assigned the task of computing inner products of the assigned coded rows with the input vector, where the assigned coded rows are random linear combinations of the rows of the original matrix. Computation time at each worker is a random variable, which is first assumed to have shifted exponential distribution, and we later generalize it to shifted Weibull distribution. The master node receives the results from the worker nodes and aggregates them until it receives a decodable set of inner products and recovers the matrix-vector multiplication. We are interested in finding the optimal load allocation that minimizes the expected waiting time to complete this computation. However, due to heterogeneity, finding the exact solution to the optimization problem seems intractable.

As the main contribution of the paper, we propose an alternative optimization that focuses on maximizing the expected number of returned computation results from the workers. Apart from being computationally tractable, the alternative optimization  asymptotically approximates the problem of finding the optimal computation load allocation. Specifically, we develop the HCMM algorithm that is derived as a solution to the alternative formulation, and prove it is asymptotically optimal. Furthermore,  we prove that given a heterogeneous cluster of $n$ workers, HCMM is $\Theta(\log n)$ times faster than uncoded schemes under the shifted exponential distribution for run-time. We further generalize the proposed HCMM algorithm to shifted Weibull model and provide similar unbounded gains over uncoded scenarios. 

In addition to proving the asymptotic optimality of HCMM, we carry out numerical studies and experiments over Amazon EC2 clusters to demonstrate how HCMM can be used in practice. We compare HCMM with three benchmark schemes -- Uniform Uncoded, Load-balanced Uncoded, and Uniform Coded. In our numerical analysis, HCMM results in significant speedups of up to $73\%$, $56\%$ and $42\%$ over the three aforementioned benchmark schemes, respectively. In experiments using Amazon EC2 clusters, we use the Luby transform (LT) codes for coding and demonstrate that HCMM combined with LT codes significantly reduces the overall execution time in comparison to uncoded and coded schemes. In particular, HCMM achieves gains of up to $61\%$, $46\%$ and $36\%$, respectively over Uniform Uncoded, Load-balanced Uncoded and Uniform Coded. Furthermore, the overall computation load of HCMM is less than the one of Uniform Coded. Our results demonstrate that HCMM combines the benefits of both Load-balanced Uncoded and Uniform Coded schemes by achieving efficient load balancing along with minimal number of redundant computations. 

Furthermore, we consider the problem of load allocation under budget constraints, considering an intuitive and convincing pricing model. In particular, we show that HCMM is the (asymptotically) optimal load allocation in feasible budget-constrained scenarios as well, and determine whether a budget-constrained computation task is feasible given a cluster of machines. We then  develop a heuristic algorithm to find the 
(sub)optimal load allocations using the proposed HCMM scheme. 
The heuristic is based on the observation that given a computation task and a set of machines, decreasing the number of fastest machines participating in HCMM results in smaller average cost.


\textbf{Notation.} We denote by $[n]$ the set $\{1,\cdots,n\}$ for any $n\in \mathbb{N}$. For non-negative sequences $g(n)$ and $h(n)$, we denote $g(n)=\mathcal{O}\big(h(n) \big)$ if there exist constants $c > 0$ and $n_0 \in \mathbb{N}$ such that $g(n) \leq c \cdot h(n)$ for all $n > n_0$; and $g(n) = \Theta\big(h(n) \big)$ if $g(n)=\mathcal{O}\big(h(n) \big)$ and $h(n)=\mathcal{O}\big(g(n) \big)$. Moreover, we write $g(n) = o\big(h(n) \big)$ if $\lim_{n \to \infty}g(n)/h(n) = 0$.

%% file: 2-prob_form_main_res.tex
In this section, we describe our computation model, the network model and the precise problem formulation. We then conclude with four theorems highlighting the main contributions of the paper.
\subsection{Computation Model}
\label{com_model}
We consider the problem of matrix-vector multiplication, in which given a matrix  $\bA \in \mathbb{R}^{r\times m}$ for some positive integers $r$ and $m$, we want to compute the output
$\mathbf{y}=\bA \mathbf{x}$ for an input vector $\mathbf{x}\in \mathbb{R}^m$. Due to limited computing power, the computation cannot be carried out at a single server and a distributed implementation is required. As an example, consider a matrix $\bA$ with an even number of rows and two computing nodes. The matrix can be divided into two  equally tall matrices $\bA_1$ and $\bA_2$, and each will be stored in a different worker node. The master node receives the input $\bx$ and broadcasts it to the two worker nodes. These nodes will then compute $\mathbf{y}_1=\bA_1 \mathbf{x}$ and $\mathbf{y}_2=\bA_2 \mathbf{x}$ locally and return their results to the master node, which combines them to obtain the intended outcome $\mathbf{y}=[\mathbf{y}_1;\mathbf{y}_2]=\bA \mathbf{x}$. 
This example also illustrates an uncoded implementation of distributed computing, in which results from all the worker nodes are required to recover the final result.

We now present the formal definition of \textit{Coded Distributed Computation}.
\begin{definition}
\label{Coded Distributed Computation}\textbf{(Coded Distributed Computation)}
The coded distributed implementation of a computation task $f_\bA(\cdot)$ is specified by:
\begin{itemize}
\item local data blocks $\left<\bA_{i}\right>_{i=1}^n$ and local computation tasks $\left<f_{\bA_i}^i(\cdot)\right>_{i=1}^n$; 

\item a decoding function that outputs $f_\bA(\cdot)$ given the results from a decodable set of local computations.    
\end{itemize}

\end{definition}
For matrix-vector multiplication tasks in particular, local data blocks  $\bA_{i}\in \mathbb{R}^{\ell_i\times m}$ are matrices consisting of \emph{coded} combinations of the rows in $\bA$, for non-negative integers $\ell_i$. To assign the computation tasks to each worker, we use random linear combinations of the $r$ rows of the matrix $\bA$, such that the master node can recover the result $\bA\bx$ from any $r$ inner products received from the worker nodes with probability 1. As an example, if worker $i$ is assigned a matrix-vector multiplication with matrix size $\ell_i \times m$,  it will compute $\ell_i$ inner products of the assigned coded rows of $\bA$ with $\bx$. The master node shall wait for the first $r$ inner products and will use them to decode the required output. In order to ensure the recovery of the output from any $r$ inner products received from the workers, we pick the computation matrix assigned to worker $i$ as $\mathbf{A}_i=\mathbf{S}_i\mathbf{A}$, where $\mathbf{S}_i\in \mathbb{R} ^{\ell_i\times r}$ is the coding matrix with i.i.d. $\mathcal{N}(0,1)$ entries. Worker $i$ computes $\mathbf{A}_i\mathbf{x}$ and returns the result to the master node. Upon receiving $r$ inner products, the aggregated results at the master will be in the form of $\mathbf{z}=\mathbf{S}_{(r)}\mathbf{A}\mathbf{x}$,  where $\mathbf{S}_{(r)}\in \mathbb{R}^{r\times r}$ is the aggregated coding matrix, and it is full-rank with probability $1$ \cite{rudelson2010non}. Therefore, the master node can recover $\mathbf{A}\mathbf{x}=\mathbf{S}_{(r)}^{-1}\mathbf{z}$ with probability $1$.\footnote{Although we consider random linear coding in our theoretical analysis, other codes such as Maximum-Distance Separable (MDS) codes and Luby transform (LT) codes are compatible with HCMM as well, given a decodable set of results at the master. For example, in the MDS case, the entries in the coding matrix $\{\mathbf{S}_i\}_{i=1}^n$ are drawn from a finite field.  Specifically, one can encode the rows of $\bA$ using an $(\sum_{i=1}^n\ell_i,r)$ MDS code and assign $\ell_i$ coded rows to the worker node $i$. The output $\mathbf{Ax}$ can be recovered from the inner products of any $r$ coded rows with the input vector $\bx$. Furthermore, to implement the ideas developed in this work, we use LT codes in our experiments over Amazon EC2 clusters.}\textsuperscript{,}\footnote{Instead of i.i.d.  Gaussian, we could use any continuous distribution for the random entries, since Schwartz-Zippel lemma ensures that such random matrix is full-rank with high probability}
\subsection{Network Model} \label{secII-B}
The network model is based on a master-worker setup illustrated in Fig.  \ref{fig_network}. The master node receives an input $\mathbf{x}$ and broadcasts it to all the workers. Each worker computes its assigned set of computations and unicasts the result to the master node. The master node aggregates the results from the worker nodes until it receives a decodable set of computations and recovers the output $\bA \mathbf{x}$.

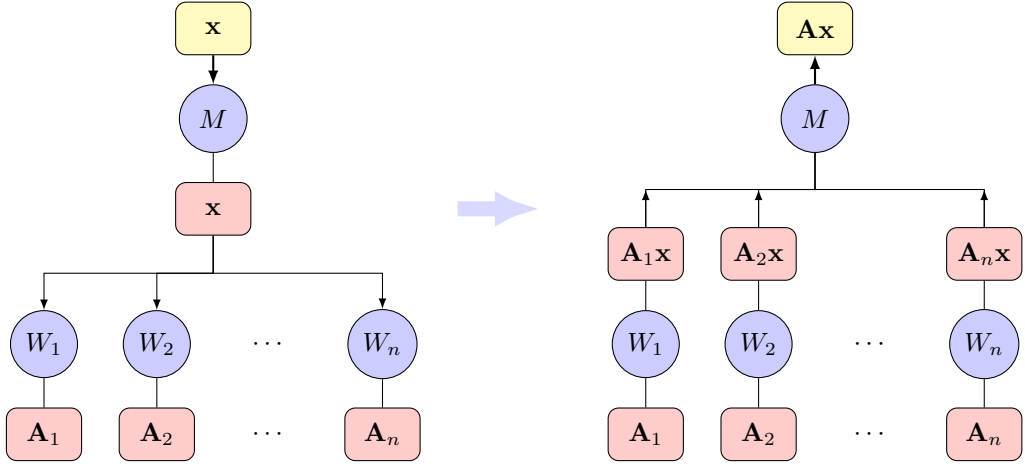
\begin{figure}[h]
\begin{center}
\begin{tikzpicture}[scale=0.68, every node/.style={transform shape}]

\node[fill=red!20,draw,rounded corners, minimum width=1cm,minimum height = .7cm] (A1) at (0,0) {${\bf{A}}_1$};

\node[fill=red!20,draw,rounded corners, minimum width=1cm,minimum height = .7cm] (A2) at (1.5,0) {${\bf{A}}_2$};

\node[fill=red!20,draw,rounded corners, minimum width=1cm,minimum height = .7cm] (An) at (4.5,0) {${\bf{A}}_n$};

\node at (3,0) {$\cdots$};

\node[fill=blue!20,draw,circle,minimum size = .7cm]  (W1) at (0,1.2) {$\mathsf{W}_1$};

\node[fill=blue!20,draw,circle,minimum size = .7cm] (W2) at (1.5,1.2) {$\mathsf{W}_2$};

\node[fill=blue!20,draw,circle,minimum size = .7cm] (Wn) at (4.5,1.2) {$\mathsf{W}_n$};

\node at (3,1.2) {$\cdots$};

\draw  (A1) to (W1);
\draw  (A2) to (W2);
\draw  (An) to (Wn);

\node[fill=red!20,draw,rounded corners, minimum width=1cm,minimum height = .7cm] (x1) at (2.25,3) {$ {\bf{x}}$};

\draw [->,arrows={-latex}]
($ (x1.south) $)
-- ++(0,-.5)
-| ($ (W1.north) $);
\draw [->,arrows={-latex}]
($ (x1.south) $)
-- ++(0,-.5)
-| ($ (W2.north) $);

\draw [->,arrows={-latex}]
($ (x1.south) $)
-- ++(0,-.5)
-| ($ (Wn.north) $);

\node[fill=blue!20,draw,circle,minimum size = .9cm]  (M1) at (2.25,4.2) {$\mathsf{M}$};

\draw  (M1) to (x1);

\node[fill=yellow!30,draw,rounded corners, minimum width=1cm,minimum height = .7cm] (x11) at (2.25,5.4) {$ {\bf{x}}$};

\draw[->,arrows={-latex},thick]  (x11) to (M1);

\node[fill=red!20,draw,rounded corners, minimum width=1cm,minimum height = .7cm] (A12) at (7,0) {${\bf{A}}_1$};

\node[fill=red!20,draw,rounded corners, minimum width=1cm,minimum height = .7cm] (A22) at (8.5,0) {${\bf{A}}_2$};

\node[fill=red!20,draw,rounded corners, minimum width=1cm,minimum height = .7cm] (An2) at (11.5,0) {${\bf{A}}_n$};

\node at (10,0) {$\cdots$};

\node[fill=blue!20,draw,circle,minimum size = .7cm]  (W12) at (7,1.2) {$\mathsf{W}_1$};

\node[fill=blue!20,draw,circle,minimum size = .7cm] (W22) at (8.5,1.2) {$\mathsf{W}_2$};

\node[fill=blue!20,draw,circle,minimum size = .7cm] (Wn2) at (11.5,1.2) {$\mathsf{W}_n$};

\node at (10,1.2) {$\cdots$};

\draw  (A12) to (W12);
\draw  (A22) to (W22);
\draw  (An2) to (Wn2);

\node[fill=red!20,draw,rounded corners, minimum width=1cm,minimum height = .7cm] (A1x) at (7,2.4) {${\bf{A}}_1 {\bf{x}}$};

\node[fill=red!20,draw,rounded corners, minimum width=1cm,minimum height = .7cm] (A2x) at (8.5,2.4) {${\bf{A}}_2 {\bf{x}}$};

\node[fill=red!20,draw,rounded corners, minimum width=1cm,minimum height = .7cm] (Anx) at (11.5,2.4) {${\bf{A}}_n {\bf{x}}$};

\draw  (A1x) to (W12);
\draw  (A2x) to (W22);
\draw  (Anx) to (Wn2);

\node[fill=blue!20,draw,circle,minimum size = .9cm]  (M2) at (9.25,4.2) {$\mathsf{M}$};

\node[fill=yellow!30,draw,rounded corners, minimum width=1cm,minimum height = .7cm] (Ax) at (9.25,5.4) {$ {\bf{Ax}}$};

\draw 
($ (A1x.north) $)
-- ++(0,.5)
-| ($ (M2.south) $);
\draw 
($ (A2x.north) $)
-- ++(0,.5)
-| ($ (M2.south) $);
\draw 
($ (Anx.north) $)
-- ++(0,.5)
-| ($ (M2.south) $);

\draw[->,arrows={-latex}] (A1x.north) -- ++ (0,.5);
\draw[->,arrows={-latex}] (A2x.north) -- ++ (0,.5);
\draw[->,arrows={-latex}] (Anx.north) -- ++ (0,.5);

\draw[->,,arrows={-latex}, thick] (M2.north) to (Ax.south);

\draw [>={LaTeX[width=4mm,length=3mm]},->,color=black!25,line width=2mm] (5cm,4cm) to (6cm,4cm);

\end{tikzpicture}
\end{center}
\caption{Master-worker setup of the computing clusters: The master node receives the input vector $\bx$ and broadcasts it to all the worker nodes. Upon receiving the input, worker node $i$ starts computing the inner products of the input vector with the locally assigned rows, i.e., $\mathbf{y}_i=\mathbf{A}_i \mathbf{x}$, and unicasts the output vector $\mathbf{y}_i$ to the master node upon completing the computation. The results are aggregated at the master node until $r$ inner products are received and the desired output $\mathbf{Ax}$ is recovered. }
\label{fig_network}
\end{figure}

We denote by $T_i$ the random variable representing the task run-time at node $i$ and assume that the run-times $T_1,\cdots,T_n$ are mutually independent. We consider the distribution of run-time random variables to be exponential, and later generalize it to Weibull distribution. More specifically, we consider a 2-parameter shifted exponential distribution for the execution time of each worker, i.e., the CDF of execution time of worker node $i$, $T_i$, loaded with $\ell_i$ row vectors is as follows:
\begin{equation}
\label{eq:1}
\Prob[T_i\leq t]=1-e^{-\frac{\mu_i}{\ell_i}(t-a_i \ell_i)},
\end{equation}
for $t\geq a_i \ell_i$ and $i \in [n]$, where $a_i>0$ is the shift parameter and $\mu_i>0$ denotes the straggling parameter associated with worker node $i$. The shifted exponential model for computation time, which is the sum of a constant (deterministic) term and a variable (stochastic) term, is motivated by the distribution model proposed by authors in \cite{liang2014tofec} for latency in querying data files from cloud storage systems. As demonstrated in \cite{lee2015speeding} as well as by our own experiments, exponential model provides a good fit for the distribution of computation times over cloud computing environments such as Amazon EC2 clusters. Moreover, these experiments confirm the assumption that as a first order approximation, both shift and mean parameters of the shifted exponential distributions linearly scale with the load size.

We further generalize the analysis to shifted Weibull distribution in Section \ref{sec:weibull}, where we consider a 3-parameter shifted Weibull distribution for the execution time of each worker. That is, the CDF of task run-time at worker node $i$, loaded with $\ell_i$ row vectors is as follows:
\begin{equation}
\label{eq:Wdist}
\Prob[ T_i \leq t] = 1-e^{-\left(\frac{\mu_i}{\ell_i}(t-a_i \ell_i)\right)^{\alpha_i}}, 
\end{equation}
for $t\geq a_i \ell_i$ and $i \in [n]$, where $a_i>0$ denotes the shift parameter, $\mu_i>0$ is the straggling parameter and $\alpha_i > 0$ represents the shape parameter associated with worker $i$. A similar model has been considered in \cite{dutta2017coded} as well.

\subsection{Problem Formulation} \label{Problem Formulation}
We consider the problem of using a cluster of $n$ worker nodes for distributedly computing the matrix-vector  multiplication $\bA{\bx}$, where $\bA$ is a size $r\times m$ matrix for positive integers $r$ and $m$. Let $\Bell=(\ell_1,\cdots,\ell_n)$ be the load allocation vector where $\ell_i$ denotes the number of rows assigned to worker node $i$. Let $T_\mathsf{CMP}$ be the random variable denoting the waiting time for receiving a decodable set of results, i.e. at least $r$ inner products. We aim at finding the optimal load allocation vector that minimizes the average waiting time by solving the following optimization problem:
\begin{equation}
\label{eq:2}
\begin{aligned}
\mathcal{P}_{\textnormal{main}}:& & \underset{{\Bell}}{\text{minimize}}
& & \Expc[T_\mathsf{CMP}].
\end{aligned}
\end{equation}
For a homogeneous cluster, to achieve a coded solution, one can divide $\mathbf{A}$ into $k$ equal size submatrices, and apply an $(n,k)$ MDS code to these submatrices. The master node can then obtain the final result from any $k$ responses. In  \cite{lee2015speeding}, the authors find the optimal $k$ for minimizing the average running time for the shifted exponential run-time model.

For heterogeneous clusters, however, assigning equal loads to servers is clearly not optimal. Moreover, directly finding the optimal solution to $\mathcal{P}_{\textnormal{main}}$ is hard. In homogeneous clusters, the problem of finding a sufficient number of inner products can be mapped to the problem of finding the waiting time for a set of fastest responses, and thus closed form expressions for the expected computation time can be found using order statistics of i.i.d. run-times. However, this is not straight-forward in heterogeneous clusters, where the load allocation is non-uniform. In Section \ref{HCMM}, we present an alternative formulation to $\mathcal{P}_{\textnormal{main}}$ in (\ref{eq:2}), and show that the solution to the alternative formulation -- which we shall name HCMM -- is tractable and provably asymptotically optimal. 

\textbf{Assumptions.} From now onward, we consider the practically relevant regime where the size of the problem scales linearly with the size of the network, while the computing power and the storage capacity of each worker node remain constant. Specifically, we assume $r=\Theta{(n)}$, $a_i=\Theta{(1)}$, $\mu_i=\Theta{(1)}$ and $\alpha_i=\Theta{(1)}$ for each worker $i$. 

\subsection{Main Results}
Having set the model and formulation of the problem, we now present the main contributions of this paper. The following theorem characterizes the asymptotic optimality of HCMM for the shifted exponential run-time model.

\begin{theorem}
\label{theorem1}
Let $T_{\mathsf{HCMM}}$ be the random variable denoting the finish time of the HCMM algorithm and $T_{\mathsf{OPT}}$ be the random variable representing the finish time of the optimum algorithm obtained by solving $\mathcal{P}_{\textnormal{main}}$. Then, for shifted exponential run-times in (\ref{eq:1}) with constant parameters $a_i=\Theta{(1)}$ and $\mu_i=\Theta{(1)}$ for each worker $i \in [n]$ and $r=\Theta{(n)}$, we have $\lim_{n \to \infty} \Expc[T_{\mathsf{HCMM}}] =\lim_{n \to \infty} \Expc[T_{\mathsf{OPT}}]$.
\end{theorem}
\begin{remark}
Theorem \ref{theorem1} demonstrates that our proposed HCMM algorithm is asymptotically optimal as the number of workers $n$ approaches infinity. In other words, the optimal computation load allocation problem $\mathcal{P}_{\textnormal{main}}$ in (\ref{eq:2}) can be optimally solved using the proposed HCMM algorithm as $n$ gets large. 
\end{remark}
\begin{remark}
We note that $\mathcal{P}_{\textnormal{main}}$ in (\ref{eq:2}) is a hard combinatorial optimization problem since it will require checking all load combinations to minimize the overall expected execution time. The key idea in Theorem \ref{theorem1} is to consider an alternative formulation to (\ref{eq:2}) focusing on maximizing the expected number of returned computation results from the workers, i.e. maximizing the \textit{aggregate return}. As we describe in Section \ref{HCMM}, the alternative optimization problem not only can be solved efficiently in a tractable way giving rise to HCMM algorithm, it also asymptotically approximates $\mathcal{P}_{\textnormal{main}}$ and allows us to establish Theorem \ref{theorem1}. 
\end{remark}

\begin{remark}
While Theorem \ref{theorem1} theoretically characterizes the optimality of our proposed scheme HCMM, we also demonstrate gains that one can get in practice. In particular, we carry out numerical studies and experiments over Amazon EC2 clusters that demonstrate that HCMM can provide significant gains in a wide variety of computing scenarios. In particular, we compare HCMM's performance with three benchmark load allocation policies -- Uniform Uncoded, Load-balanced Uncoded, and Uniform Coded. In numerical studies, HCMM achieves speedups of up to $71\%$ over Uniform Uncoded, up to $53\%$ over Load-balanced Uncoded, and up to $39\%$ over Uniform Coded. In EC2 experiments, HCMM combined with the Luby transform (LT) codes  provides speedups of up to $61\%$, $46\%$ and $36\%$  over Uniform Uncoded, Load-balanced Uncoded and Uniform Coded, respectively.
\end{remark}


\begin{theorem}
\label{theorem2}
Let $T_{\mathsf{UC}}$ denote the completion time of the uncoded distributed matrix multiplication algorithm. Then, for the shifted exponential run-times with constant parameters and $r=\Theta{(n)}$,
\begin{equation}
\frac{\Expc[T_{\mathsf{UC}}]}{\Expc[T_{\mathsf{HCMM}}] }=\Theta\big(\log n\big). \nonumber
\end{equation}
\end{theorem}

\begin{remark}
As Theorem \ref{theorem2} shows, our proposed HCMM guarantees an improvement of $\Theta\big(\log n\big)$ in expected execution time over \emph{any} uncoded scheme, including the one that optimally allocates the workers' loads. This result illustrates that by leveraging coded computing, one achieves the same order-wise gain over heterogeneous clusters as over homogeneous clusters \cite{lee2015speeding}. 
\end{remark}

Although Theorems \ref{theorem1} and \ref{theorem2} are based on the shifted exponential model (\ref{eq:1}) for run-time random variables for the workers, our analyses are general and can be extended to other models. The following two theorems generalize the results when the execution time of each worker follows the Weibull distribution as described in (\ref{eq:Wdist}).
\begin{theorem}
\label{theorem3}
 For the shifted Weibull distribution of run-times with constant parameters $a_i=\Theta{(1)}$, $\mu_i=\Theta{(1)}$ and $\alpha_i=\Theta{(1)}$ for each worker $i \in [n]$ and $r=\Theta{(n)}$, the proposed HCMM algorithm is asymptotically optimal, i.e., $\lim_{n \to \infty} \Expc[T_{\mathsf{HCMM}}] =\lim_{n \to \infty} \Expc[T_{\mathsf{OPT}}]$.
\end{theorem}
\begin{theorem}
\label{theorem4}
 Under the Weibull distribution for run-times with constant parameters and $r=\Theta{(n)}$, the proposed HCMM scheme unboundedly outperforms the uncoded scheme, i.e.,
\begin{equation}
\frac{\Expc[T_{\mathsf{UC}}]}{\Expc[T_{\mathsf{HCMM}}] } \geq \Theta\big( (\log n)^{1/\tilde{\alpha}}\big),  \nonumber
\end{equation}
where $\tilde{\alpha} = \max_{i \in [n]} \alpha_i$ is the largest shape parameter among the workers.
\end{theorem}

\begin{remark}
As stated in Theorem \ref{theorem4}, HCMM provides an unbounded gain over \emph{any} uncoded scheme -- including the optimal uncoded load allocation -- under the Weibull distribution for workers' run-times. Furthermore, our numerical simulations demonstrate speedups of up to $73\%$, $56\%$ and $42\%$ over Uniform Uncoded, Load-balanced Uncoded and Uniform Coded, respectively.
\end{remark}

In the following section, we describe our alternative formulation based on aggregate return and describe our proposed HCMM algorithm that solves the alternative optimization.

%% file: 3-HCMM_scheme.tex
In this section, we prove Theorems \ref{theorem1} and \ref{theorem2} for the exponential model (\ref{eq:1}). In particular, we start by describing the HCMM algorithm and show that it asymptotically achieves the optimal performance, as stated in Theorem \ref{theorem1}, and lastly conclude the section by characterizing the gain of HCMM over uncoded scheme.

To derive HCMM, we start by reformulating $\mathcal{P}_{\textnormal{main}}$ defined in (\ref{eq:2}) and show that the alternative formulation can be efficiently solved, as opposed to solving $\mathcal{P}_{\textnormal{main}}$ that needs an exhaustive search over all possible load allocations. The solution to the alternative problem gives rise to HCMM. We will further prove the optimality of HCMM and compare its average run-time to uncoded schemes. 
\subsection{Alternative Formulation of $\mathcal{P}_{\textnormal{main}}$ via Maximal Aggregate Return}

Consider an $n$-tuple load allocation $\Bell=(\ell_1,\cdots,\ell_n)$ and let $t$ be a feasible time for computation, i.e., $t\geq\displaystyle \max_{i}\{a_i\,\ell_i\}$. The number of equations received from worker $i\in [n]$ at the master node till time $t$ is a random variable, $X_i(t)=\ell_i\, \mathbbm{1}_{\{T_i\leq t\}}$, where $T_i$ is the random execution time for machine $i$ that is assigned the load $\ell_i$ and $\mathbbm{1}_{\{\cdot\}}$ denotes the indicator function. Then, the \textit{aggregate return} at the master node at time $t$ is:
\begin{equation}
\label{aggrret}
X(t)=\sum_{i=1}^{n}X_i(t).\nonumber
\end{equation}

We propose the following two-step alternative formulation for $\mathcal{P}_{\textnormal{main}}$ defined in (\ref{eq:2}). First, for a fixed feasible time $t$, we maximize the aggregate return over different load allocations, i.e., we solve
\begin{align}
\label{eq:4}
& \mathcal{P}_{\textnormal{alt}}^{(1)}: \Bell^*(t) =\arg \max_{\Bell} \Expc\big[X(t)\big].
\end{align}
Then, given the load allocation $\Bell^*(t)=\big(\ell^*_1(t),\cdots,\ell^*_n(t) \big)$ obtained from $\mathcal{P}_{\textnormal{alt}}^{(1)}$, we find the smallest time $t$ such that with high probability, there is enough aggregate return by time $t$ at the master node, i.e., we solve
\begin{equation}
\label{eq:5}
\begin{aligned}
 \mathcal{P}_{\textnormal{alt}}^{(2)}:\,\,\,\,\, &{\text{minimize}}
&& t \\
&  \text{subject to}
&& \Prob\big[{X^*}(t)<r\big] = o\left(\frac{1}{ n}\right),
\end{aligned}
\end{equation}
where ${X^*}(t)$ is the aggregate return at time $t$ for  load allocation obtained from $\mathcal{P}_{\textnormal{alt}}^{(1)}$, that is
\begin{equation}
\label{aggreturn}
X^*(t)=\sum_{i=1}^{n} X_i^*(t)=\sum_{i=1}^{n}\ell^*_i(t)\,\mathbbm{1}_{\{T_i \leq t\}}. \nonumber
\end{equation}
From now onward, we denote the solution to $\mathcal{P}_{\textnormal{alt}}^{(2)}$ by $t^*$ and hence  $\Bell^*(t^*)$ denotes the solution to the two-step alternative formulation in (\ref{eq:4}) and (\ref{eq:5}) which gives rise to our proposed HCMM scheme described next. 

\subsection{Solving the Alternative Formulation}\label{secIII-B}

Considering the exponential distribution for workers' run-times, we first proceed to solve $\mathcal{P}_{\textnormal{alt}}^{(1)}$ in  (\ref{eq:4}). The expected number of equations aggregated at the master node at time $t$ is:
\begin{equation}
\label{eq:3}
\Expc\big[X(t)\big]=\sum_{i=1}^{n}\Expc\big[X_i(t)\big]=\displaystyle \sum_{i=1}^n \ell_i\left(1-e^{-\frac{\mu_i}{\ell_i}(t-a_i\,\ell_i)}\right). \nonumber
\end{equation}
Since there is no constraint on load allocations, $\mathcal{P}_{\textnormal{alt}}^{(1)}$ can be decomposed to $n$ decoupled optimization problems, i.e.,
\begin{equation}
\label{liopt}
\ell^*_i(t)=\arg \max_{\ell_i} {\Expc\big[X_i(t)\big]}, 
\end{equation}
for all workers $i \in [n]$. The solution to (\ref{liopt}) satisfies the following optimality condition:
\begin{align}
\frac{\partial }{\partial \ell_i}\Expc\left[X_i(t)\right] &=1-e^{-\frac{\mu_i}{\ell_i}(t-a_i \ell_i)}\left(\frac{\mu_i t}{\ell_i}+1\right)=0, \nonumber
\end{align}
which yields 
\begin{equation}\label{optload}
    \ell_i^*(t) =\frac{t}{\lambda_i}, 
\end{equation}
where $\lambda_i=\Theta{(1)}$ is a constant independent of $t$ and is the positive solution to the following equation:
\begin{equation}
\label{eq:lambda}
e^{\mu_i \lambda_i}=e^{a_i \mu_i}(\mu_i \lambda_i+1). \nonumber
\end{equation}
One can easily check that the condition $t \geq a_i \ell_i^*(t)$ holds for all $i$ as well.  Moreover, we denote by $t^*$ the solution to $\mathcal{P}_{\textnormal{alt}}^{(2)}$. Now, we define the HCMM load allocation as 
\begin{equation}\label{eq:HCMM_load}
    \ell_i^*(t^*) =\frac{t^*}{\lambda_i},
\end{equation}
for all workers $i$.  In the following,  we  formally define the HCMM algorithm  which is basically the solution to 
$\mathcal{P}_{\textnormal{alt}}$.

 \begin{algorithm}[t]
\caption{Heterogeneous Coded Matrix Multiplication (HCMM)}  \label{def:HCMM}
\hspace*{\algorithmicindent} \textbf{Input:} computation time parameters $(a_i, \mu_i)$ for each \\
\hspace*{\algorithmicindent} worker $i$ \footnotemark  \\
\hspace*{\algorithmicindent} \textbf{Output:} computation load assigned to each worker $i$
\begin{algorithmic}[1]
\Procedure{HCMM}{}
\State solve $\mathcal{P}_{\textnormal{alt}}^{(1)}$ for any feasible $t$
\State obtain $\ell_i^*(t) =\frac{t}{\lambda_i}$
\State solve $\mathcal{P}_{\textnormal{alt}}^{(2)}$ and obtain $t^*$
\State \Return  $\ell_i^*(t^*) = \frac{t^*}{\lambda_i}$ row vector computations for\\
\hspace*{\algorithmicindent}worker $i$
\EndProcedure
\end{algorithmic}
\end{algorithm}

\begin{remark}
We note that in order to implement any load allocation scheme, each worker supposedly admits an integer number of rows as its associated computation load. However, the load allocation $\ell_i^*(t^*)$ given by HCMM scheme in Algorithm \ref{def:HCMM} is a real number for any worker $i$ and therefore one needs to round the result before proceeding with experiments. In practical scenarios, $\ell_i^*(t^*)$ is fairly large, e.g.  in the order of $100$ row vectors. Therefore, the effect of rounding the load allocations shall be insignificant. 
\end{remark}

We now provide an approximation to $t^*$ and show it asymptotically converges to $t^*$. The expected aggregate return at time $t$ for optimal loads obtained in (\ref{optload}) is
\begin{align}
\label{eq:12}\Expc\left[X^*(t)\right] &= \sum_{i=1}^n \ell_i^*(t)\left(1-e^{-\frac{\mu_i}{\ell_i^*(t)}\left(t-a_i \ell_i^*(t)\right)}\right) \nonumber\\
&= \sum_{i=1}^n \frac{t}{\lambda_i}\left(1-e^{-\frac{\mu_i}{{t}/{\lambda_i}}\left(t- \frac{a_i t}{\lambda_i}\right)}\right) \nonumber\\
&=t s,
\end{align}
where 
\begin{equation} s=\sum_{i=1}^{n}\frac{1}{\lambda_i}\left(1-e^{-\mu_i\lambda_i(1-\frac{a_i} {\lambda_i})}\right) =\sum_{i=1}^{n} \frac{\mu_i}{1+\mu_i\lambda_i}= \Theta(n), \nonumber
\end{equation}
since $\mu_i=\Theta{(1)}$ and $\lambda_i=\Theta{(1)}$. Let $\tau^*$ be the solution to the following equation when solved for $t$: 
\begin{equation}
\label{eq:8}
\Expc\big[X^*(t)\big]=\sum_{i=1}^n \ell_i^*(t)\left(1-e^{-\frac{\mu_i}{\ell_i^*(t)}(t-a_i \ell_i^*(t))}\right)=r.
\end{equation}
In other words, $\tau^*$ is the time for which there are exactly $r$ inner products -- on average -- aggregated at the master node, when the workers are loaded according to the loading obtained in (\ref{optload}).
Using (\ref{optload}), (\ref{eq:12}) and (\ref{eq:8}), we find that 
\begin{align}
\tau^* &= \frac{r}{s}=\Theta(1),\label{topt}\\
\ell_i^*(\tau^*) &= \frac{\tau^*}{\lambda_i}=\frac{r}{s\lambda_i}=\Theta(1).  \label{lopt}
\end{align}

\footnotetext{For the shifted Weibull distribution, parameters $(a_i, \mu_i, \alpha_i)$ are taken as inputs.}

We now present the following lemma, which shows that $\tau^*$ converges to $t^*$ for large $n$ (see  Appendix for proof).

\begin{lemma}
\label{lemma1}
Let $t^*$ be the solution to the alternative formulation $\mathcal{P}_{\textnormal{alt}}$ in (\ref{eq:4}-\ref{eq:5}) and $\tau^*$ be the solution to (\ref{eq:8}). Then, 
\begin{equation}
\label{ttau}
  \tau^*\leq t^*\leq \tau^*+o(1). \nonumber
\end{equation}
\end{lemma}


\subsection{Asymptotic Optimality of HCMM}\label{secIII-C}
In this subsection, we prove the asymptotic optimality of HCMM as claimed in Theorem \ref{theorem1}.
\begin{proof}[Proof of Theorem \ref{theorem1}]
Consider the HCMM load assignment in (\ref{eq:HCMM_load}). Let the random variable $T_{\mathsf{HCMM}}$ denote the finish time associated to this load allocation, i.e. the waiting time to receive at least $r$ inner products from the workers. Let $T_{\mathsf{max}}$ be the random variable denoting the finish time of all the workers for the HCMM load assignment. 

First, we show that
\begin{equation}
\label{eq:16}
    \Expc[T_{\mathsf{HCMM}}]\leq t^*+o(1).\nonumber
\end{equation}
Let us define two events $\mathcal{E}_1$ and $\mathcal{E}_2$ as follows:
\begin{equation}
\label{eq:events}
    \mathcal{E}_1 = \{ T_{\mathsf{max}} > \Theta(n)\}  \, \text{ and } \, \mathcal{E}_2 = \{ T_{\mathsf{HCMM}} > t^*\}.\nonumber
\end{equation}
Conditioning on these events, we can write 
\begin{align}
\label{eq:th1}
    \Expc[T_{\mathsf{HCMM}}] &= \Expc[T_{\mathsf{HCMM}} | \mathcal{E}_1] \Prob[\mathcal{E}_1 ] \nonumber\\
    & \quad + \Expc[T_{\mathsf{HCMM}} | {\mathcal{E}_1}^c \cap \mathcal{E}_2] \Prob[{\mathcal{E}_1}^c \cap \mathcal{E}_2 ] \nonumber\\
    & \quad + \Expc[T_{\mathsf{HCMM}} | {\mathcal{E}_1}^c \cap {\mathcal{E}_2}^c] \Prob[{\mathcal{E}_1}^c \cap {\mathcal{E}_2}^c ].
\end{align}
We can write the second term in RHS of (\ref{eq:th1}) as follows:
\begin{align}
\label{eq:th2}
    \Expc[ & T_{\mathsf{HCMM}} | {\mathcal{E}_1}^c \cap \mathcal{E}_2] \Prob[{\mathcal{E}_1}^c \cap \mathcal{E}_2 ] \nonumber\\
    & = \Expc[T_{\mathsf{HCMM}} | T_{\mathsf{max}} \leq \Theta(n) , T_{\mathsf{HCMM}} > t^*] \nonumber\\
    & \quad \times \Prob[T_{\mathsf{max}} \leq \Theta(n) , T_{\mathsf{HCMM}} > t^* ]  \nonumber\\
    & \leq \Expc[T_{\mathsf{max}} | T_{\mathsf{max}} \leq \Theta(n) , T_{\mathsf{HCMM}} > t^*] \Prob[ T_{\mathsf{HCMM}} > t^* ]  \nonumber\\
    &\overset{(a)}{\leq} \Theta(n) \cdot o\left(\frac{1}{ n}\right)\nonumber\\
    &= o(1).
\end{align}
To prove $(a)$, we note that HCMM returns $r$ inner products by time $T_{\mathsf{HCMM}}$. Moreover, the aggregate return is increasing in time. Therefore, 
\begin{align}
    \Prob[ T_{\mathsf{HCMM}} > t^* ]  \leq \Prob[ X^*(t^*) < r ] = o\left( \frac{1}{n} \right).\nonumber
\end{align}
Furthermore, we have
\begin{align}
    &\Expc[T_{\mathsf{max}} | T_{\mathsf{max}} \leq \Theta(n) , T_{\mathsf{HCMM}} > t^*] \nonumber\\
    & \quad = \frac{1}{\Pr[T_{\mathsf{max}} \leq \Theta(n) , T_{\mathsf{HCMM}} > t^*]} \nonumber\\
    & \quad \quad \times \int_{t_1 = 0}^{\Theta(n) } \int_{t_2 = t^*}^{\infty} t_1 d \Pr[T_{\mathsf{max}}  \leq t_1, T_{\mathsf{HCMM}} \leq t_2] \nonumber\\
    & \quad \leq \frac{\Theta(n)}{\Pr[T_{\mathsf{max}} \leq \Theta(n) , T_{\mathsf{HCMM}} > t^*]} \nonumber\\
    & \quad \quad \times \int_{t_1 = 0}^{\Theta(n) } \int_{t_2 = t^*}^{\infty}  d \Pr[T_{\mathsf{max}}  \leq t_1, T_{\mathsf{HCMM}} \leq t_2] \nonumber\\
    & \quad = \Theta(n).\nonumber
\end{align}
Moreover, the third term in RHS of (\ref{eq:th1}) can be written as 
\begin{align}
\label{eq:th3}
    \Expc[ T_{\mathsf{HCMM}} | {\mathcal{E}_1}^c & \cap {\mathcal{E}_2}^c]  \Prob[{\mathcal{E}_1}^c \cap {\mathcal{E}_2}^c ] \nonumber\\
    & = \Expc[T_{\mathsf{HCMM}} | T_{\mathsf{max}} \leq \Theta(n) , T_{\mathsf{HCMM}} \leq  t^*] \nonumber\\
    & \quad \times\Prob[T_{\mathsf{max}} \leq \Theta(n) , T_{\mathsf{HCMM}} \leq  t^* ]  \nonumber\\
    & \leq  \Expc[T_{\mathsf{HCMM}} | T_{\mathsf{max}} \leq \Theta(n) , T_{\mathsf{HCMM}} \leq  t^*]\nonumber\\
    & \overset{(b)}{\leq} t^*,
\end{align}
where proof of $(b)$ is similar to proof of $(a)$ in (\ref{eq:th2}). Regarding the first term in RHS of (\ref{eq:th1}), we have 
\begin{align}
\label{eq:th4}
    \Expc[T_{\mathsf{HCMM}} | \mathcal{E}_1] \Prob[\mathcal{E}_1 ]  & = \Expc[T_{\mathsf{HCMM}} | T_{\mathsf{max}} > \Theta(n)] \nonumber\\
    & \quad \times \Prob[T_{\mathsf{max}} > \Theta(n) ]  \nonumber\\
    & \leq \Expc[T_{\mathsf{max}} | T_{\mathsf{max}} > \Theta(n)] \nonumber\\
    & \quad \times \Prob[T_{\mathsf{max}} > \Theta(n) ] \nonumber\\
    &= \int_{\Theta(n)}^{\infty}tf_{\mathsf{max}}(t)\,dt \nonumber\\
     & \overset{(c)} \leq \int_{\Theta(n)}^{\infty}tnk_1 e^{-k_1t}\left(1-e^{-k_1t}\right)^{n-1} \,dt \nonumber\\
    &\leq   \int_{\Theta(n)}^{\infty} n k_1 t e^{-k_1t} \,dt \nonumber\\
   &\leq \int_{\Theta(n)}^{\infty} \frac{1}{t^2} \,dt = o(1),
\end{align}
for some $k_1 = \Theta(1)$ and large enough $n$. To derive inequality $(c)$, we find a stochastic upper bound on $T_{\mathsf{max}}$ by considering $n$ i.i.d. copies of the worker run-times with largest shift and smallest straggling parameters that are also $\Theta(1)$, and use the PDF of the maximum of $n$ i.i.d. exponential random variables. As we later use in the proof of Theorem \ref{theorem3}, one can similarly write  for the shifted Weibull distribution: 
\begin{align}
\label{eq:th5}
    \Expc[T_{\mathsf{HCMM}} | \mathcal{E}_1] & \Prob[\mathcal{E}_1 ]  \nonumber\\
    & \leq \int_{\Theta(n)}^{\infty}tf_{\mathsf{max}}(t)\,dt \nonumber\\
    &\leq \int_{\Theta(n)}^{\infty} nk_1 k_2 t^{k_2} e^{-k_1t^{k_2}}\left(1-e^{-k_1t^{k_2}}\right)^{n-1} \,dt \nonumber\\
    & \leq\int_{\Theta(n)}^{\infty}  nk_1 k_2 t^{k_2} e^{-k_1t^{k_2}} \,dt \nonumber\\
    &\leq \int_{\Theta(n)}^{\infty} \frac{1}{t^2} \,dt = o(1),
\end{align}
for some constants $k_1$ and $k_2$. Therefore, using (\ref{eq:th2}), (\ref{eq:th3}) and (\ref{eq:th4}) (or (\ref{eq:th5}) for the shifted Weibull model) in (\ref{eq:th1}) we have
\begin{equation}
\Expc[T_{\mathsf{HCMM}}]\leq t^*+o(1).\nonumber
\end{equation}

Let $\Bell_{\mathsf{OPT}}=(\ell_{\mathsf{OPT},1},\cdots,\ell_{\mathsf{OPT},n})$ denote the optimal load allocation corresponding to $\mathcal{P}_{\textnormal{main}}$ in (\ref{eq:2}) and $X_{\mathsf{OPT}}(\cdot)$ represent the aggregate return under load allocation $\Bell_{\mathsf{OPT}}$. Now we prove the following lower bound on the average completion time of the optimum algorithm:
\begin{equation}
\Expc[T_{\mathsf{OPT}}]  \geq    t^*-o(1).\nonumber
\end{equation}
 To this end, we show the following two inequalities,
\begin{equation}
\Expc[T_{\mathsf{OPT}}] \overset{(d)} \geq \tau -\delta_1  \overset{(e)} \geq t^*-\delta_2-\delta_1,\nonumber
\end{equation}
where $\delta_1 = \Theta \left(\frac{\log n}{\sqrt{n}}\right) $, $\delta_2=\Theta \left(\frac{\log n}{\sqrt{n}}\right)$ and $\tau$ is the solution to $
 \Expc[X_{\mathsf{OPT}}(\tau)]=r.$
We have
\begin{align}
r-\Expc[X_{\mathsf{OPT}} & (\tau-\delta_1)] \nonumber\\
&= \sum_{i=1}^{n} \ell_{\mathsf{OPT},i} \big( \Prob[T_i<\tau]-\Prob[T_i<\tau-\delta_1]   \big)\nonumber \\
&= \sum_{i=1}^{n} \ell_{\mathsf{OPT},i} \left( \frac{d}{d\tau}\Prob[T_i<\tau]\delta_1+\mathcal{O}\left(\delta_1^2\right) \right)\nonumber \\
&= \Theta(n\delta_1)+\mathcal{O}\left(n \delta^2_1\right)=\Theta(n\delta_1)\label{eq:orderndelta},\nonumber
\end{align}
where we used the fact that $\ell_{\mathsf{OPT},i}=\Theta(1)$\footnote{We argue that the allocated loads in the optimum coded scheme are all $\Theta(1)$. Without loss of generality, suppose $\ell_{\mathsf{OPT},1}> \Theta(1)$ which implies $\lim_{n\to \infty}\Prob[T_1<t]=0$ for any $t=\Theta(1)$. We have already implemented HCMM, a (sub-)optimal algorithm achieving computation time $\tau^*=\Theta(1)$, therefore the optimal scheme should have a better finishing time $\tau \leq \Theta(1)$. Now assume the load of machine 1 is replaced by $\tilde{\ell}_{\mathsf{OPT},1}=\Theta(1)$. Clearly, for any time $t=\Theta(1)$, the aggregate return for the new set of loads is larger than the former one by any $\Theta(1)$ time, almost surely. This is in contradiction to optimality assumption.}.
By McDiarmid's inequality (see Appendix for its description), we have

\begin{align*}
\Prob [X_{\mathsf{OPT}}(\tau-\delta_1) & \geq  r ] \\
&  = \Prob[X_{\mathsf{OPT}}(\tau-\delta_1)-\Expc[X_{\mathsf{OPT}}(\tau-\delta_1)]\\
& \quad \quad \quad \geq r-\Expc[X_{\mathsf{OPT}}(\tau-\delta_1)]]\\
&\leq \exp \left(-\frac{2\big(\Expc[X_{\mathsf{OPT}}(\tau-\delta_1)] - r\big)^2}{\sum_{i=1}^{n} \ell^2_{\mathsf{OPT},i}}\right) \\
&= e^{-\Theta\left({n\delta_1^2}\right)}= o\left(\frac{1}{ n}\right),
\end{align*}
which implies inequality $(d)$. We proceed to prove $(e)$ by showing the following two inequalities,
\begin{equation}
\tau  \geq \tau^*, \label{eq46} 
\end{equation}
\begin{equation}
\tau^*  \geq t^*-\delta_2, \label{eq45}
\end{equation}
where $\tau^*$ is obtained in (\ref{topt}). Given the fact that HCMM maximizes the expected aggregate return, we have
\begin{equation}
\Expc[X^*(t)] \geq \Expc[X_{\mathsf{OPT}}(t)],\nonumber
\end{equation}
for every feasible $t$, which implies (\ref{eq46}). Moreover, Lemma \ref{lemma1} proves (\ref{eq45}).
All in all, we have
\begin{equation}
t^*-o(1) \leq \Expc[T_{\mathsf{OPT}}] \leq  \Expc[T_{\mathsf{HCMM}}] \leq t^* + o(1),\nonumber
\end{equation}
which yields $\lim_{n \to \infty} \Expc[T_{\mathsf{HCMM}}] =\lim_{n \to \infty} \Expc[T_{\mathsf{OPT}}]$ and the claim is concluded.
\end{proof}
\vspace{-2mm}
\subsection{Comparison with Uncoded Schemes}\label{secIII-D}
\label{Comparison with uncoded scheme}
This subsection provides the proof of Theorem \ref{theorem2} by comparing the performance of HCMM to  uncoded scheme. In an uncoded scheme, the redundancy factor is $1$; thus, the master node has to wait for the results from all the worker nodes in order to complete the computation. 

\begin{proof}[Proof of Theorem \ref{theorem2}]
We start by characterizing the expected run-time of the best uncoded scheme. Particularly, we show that 
\begin{equation}
\Expc[T_{\mathsf{UC}}]=\Theta\big(\log n\big),\nonumber
\end{equation}
where $T_{\mathsf{UC}}$ denotes the completion time of the optimum uncoded distributed matrix multiplication algorithm. To do so, we start by showing that 
\begin{equation}
\Expc[T_{\mathsf{UC}}] \geq c\log n,\nonumber
\end{equation}
for a constant $c$ independent of $n$. For a set of machines with parameters $\{(a_i,\mu_i)\}_{i=1}^n$, let $\tilde{a} = \min_i a_i$ and $\tilde{\mu} = \max_i \mu_i$.
Now, consider another set of $n$ machines in which every machine is replaced with a faster machine with parameters $(\tilde{a},\tilde{\mu})$. Since the computation times of the new set of machines are i.i.d., one can show that the optimal load allocation for these machines is uniform, i.e., 
\begin{equation}
\widetilde{\ell}^*_i=\frac{r}{n},\nonumber
\end{equation}
for every machine $i\in [n]$.
Let $\{\widetilde{T}_i\}_{i=1}^{n}$ represent the i.i.d. shifted exponential random variables denoting the execution times for the new set of machines where each machine is loaded by $\widetilde{\ell}^*_i=\frac{r}{n}$. Therefore, the CDF of the completion time of each new machine can be written as
\begin{align}
\Prob\left[\widetilde{T}_i\leq t\right]=&1-e^{{-\frac{\tilde{\mu}}{\widetilde{\ell}^*_i}\left(t-\tilde{a} \widetilde{\ell}^*_i\right)}} \nonumber\\=&1-e^{-\tilde{\mu}\frac{n}{r}\left(t-\tilde{a} \frac{r}{n}\right)},\nonumber
\end{align}
for $t\geq  \frac{\tilde{a} r}{n}$ and the expected computation time can be written as
\begin{equation}
\Expc\left[\widetilde{T}_i\right]=\frac{r}{n}\left(\tilde{a} + \frac{1}{\tilde{\mu}}\right),\nonumber
\end{equation}
for all $i\in [n]$. Since the master needs to wait for all of the machines to return their results, the total run-time is $\widetilde{T}_{\mathsf{UC}} = \max_{i \in [n]} \widetilde{T}_i$. Therefore,  
\begin{equation}\label{eq:best_T_UC}
    \Expc\left[\widetilde{T}_{\mathsf{UC}}\right] = \Expc \left[\max_{i \in [n]} \widetilde{T}_i \right]=\frac{\tilde{a} r}{n}+\frac{r H_n}{n \tilde{\mu}},
\end{equation}
where $H_n=1+\frac{1}{2}+\frac{1}{3}+\cdots+\frac{1}{n}$ is the sum of the harmonic series. We can further bound (\ref{eq:best_T_UC}) using the fact that
\begin{equation}
    \frac{\tilde{a} r}{n}+\frac{r H_n}{n \tilde{\mu}} \geq \frac{\tilde{a} r}{n}+\frac{r}{n \tilde{\mu}}\log(n+1) \geq c \log n,\nonumber
\end{equation}
for a constant $c$ independent of $n$, since $r=\Theta(n)$, $\tilde{a}=\Theta(1)$, and $\tilde{\mu}=\Theta(1)$ for all $i \in [n]$. All in all, we have the following lower bound on the optimal uncoded scheme:
\begin{equation}\label{lowerb}
\Expc[T_{\mathsf{UC}}] \geq \Expc\left[\widetilde{T}_{\mathsf{UC}}\right] \geq c \log n.
\end{equation}

Now consider another set of $n$ machines, where each machine is replaced with a slower one with parameters $(\hat{a},\hat{\mu})$ for $\hat{a} = \max_i a_i$ and $\hat{\mu} = \min_i \mu_i$. By an argument similar to the one employed the lower bound, we can write  
\begin{equation}
\label{upperb}
\Expc[T_{\mathsf{UC}}] \leq \frac{\hat{a} r}{n}+\frac{r}{n \hat{\mu}}H_n \leq C \log n,
\end{equation}
for another constant $C$. From (\ref{lowerb}) and (\ref{upperb}), one can conclude that
\begin{equation}
\label{UCorder}
\Expc[T_{\mathsf{UC}}]=\Theta\big(\log n\big).
\end{equation}
Further, by Theorem \ref{theorem1} and Lemma \ref{lemma1}, we find that 
\begin{equation}
\label{HCMMorder}
\Expc[T_{\mathsf{HCMM}}]=\Theta(1).
\end{equation}
Comparing (\ref{UCorder}) to (\ref{HCMMorder}) demonstrates that HCMM outperforms the best uncoded scheme by a factor of $\Theta(\log n)$, i.e.,
\begin{equation}
\frac{\Expc[T_{\mathsf{UC}}]}{\Expc[T_{\mathsf{HCMM}}] }=\Theta\big(\log n\big).\nonumber
\end{equation}
\end{proof}

%% file: 4-gen_weibull.tex
In this section, we consider the shifted Weibull distribution for the workers' execution times, which captures a broader class of run-time models than the exponential distribution. We particularly generalize our proposed HCMM algorithm to the class of shifted Weibull distributed run-times and prove Theorems \ref{theorem3} and \ref{theorem4}. More specifically, we argue that asymptotic optimality of HCMM is derived similar to the shifted exponential case and further show that HCMM provides unbounded gain over uncoded schemes, asymptotically. 

A random variable $T$ has Weibull distribution with shape parameter $\alpha>0$ and scale parameter $\mu>0$, denoted by $T \sim \mathcal{W}(\alpha,\mu)$, if the CDF of $T$ is of the following form:
\begin{equation}
\Prob[ T \leq t] = 1-e^{-(\mu t)^\alpha}, \quad t\geq 0.\nonumber
\end{equation}
The expected value of the Weibull distribution is known to be $\Expc[T]=\frac{1}{\mu}\Gamma(1+1/\alpha)$, where $\Gamma(\cdot)$ denotes the Gamma function.

As stated in Section \ref{secII-B}, we consider a 3-parameter shifted Weibull distribution for workers' run-times defined in (\ref{eq:Wdist}). The mean value of the worker $i$'s run-times is then $\Expc[T_i]=a_i \ell_i + \frac{\ell_i}{\mu_i}\Gamma(1+1/\alpha_i)$. Clearly, shifted exponential distribution is a special case of the shifted Weibull model when $\alpha_i=1$. By slight reparameterizations, this model can be similarly applied to the HCMM algorithm proposed in Algorithm \ref{def:HCMM}, meaning that the main and alternative optimization problems defined in (\ref{eq:2}), (\ref{eq:4}) and (\ref{eq:5}) can be similarly analyzed  under the shifted Weibull model.

As in the exponential case, we begin by maximizing the expected aggregate return at the master node ($\mathcal{P}_{\textnormal{alt}}^{(1)}$) under the shifted Weibull distribution, which is given by 
\begin{equation}
\Expc\left[X(t)\right]=\sum_{i=1}^{n}\Expc\left[X_i(t)\right]=\displaystyle \sum_{i=1}^n \ell_i\left(1-e^{-\left(\frac{\mu_i}{\ell_i}(t-a_i \ell_i)\right)^{\alpha_i}}\right).\nonumber
\end{equation}
The optimal load allocation that maximizes the individual expected aggregate returns at each worker (and thus the total aggregate return) can be found by solving the following equation:
\begin{align}
&\frac{\partial }{\partial \ell_i}\Expc\left[X_i(t)\right] \nonumber\\ 
& \quad =1-e^{-\left(\frac{\mu_i}{\ell_i}(t-a_i \ell_i)\right)^{\alpha_i}}\left( 1+\frac{{\mu_i}^{\alpha_i} \alpha_i t  }{\ell_i}\left(\frac{t}{\ell_i} - a_i\right)^{\alpha_i-1}\right) \nonumber\\
& \quad =0.
\label{optload2}
\end{align}
Solving (\ref{optload2}) for $\ell_i$ yields $\ell_i^*(t)=\frac{t}{\lambda_i}$ where the constant $\lambda_i > a_i$ is the positive solution to 
\begin{equation}
e^{{\mu_i}^{\alpha_i}(\lambda_i-a_i)^{\alpha_i}} = 1 + \alpha_i {\mu_i}^{\alpha_i} \lambda_i(\lambda_i-a_i)^{\alpha_i-1}.\nonumber
\end{equation}
Similar to Section \ref{HCMM}, we can define $s$ as follows,
\begin{align}
s &= \frac{\Expc\left[X^*(t)\right]}{t} \nonumber\\
&= \frac{1}{t} \displaystyle \sum_{i=1}^n \ell_i^*(t)\left(1-e^{-\left(\frac{\mu_i}{\ell_i^*(t)}\left(t-a_i \ell_i^*(t)\right)\right)^{\alpha_i}}\right)\nonumber\\
&= \sum_{i=1}^{n}\frac{1}{\lambda_i} \left(1-e^{-\left(\mu_i\lambda_i\left(1-\frac{a_i} {\lambda_i}\right)\right)^{\alpha_i}}\right)\nonumber\nonumber\\
&= \sum_{i=1}^{n} \frac{\alpha_i {\mu_i}^{\alpha_i} (\lambda_i-a_i)^{\alpha_i-1} }{1 + \alpha_i {\mu_i}^{\alpha_i} \lambda_i(\lambda_i-a_i)^{\alpha_i-1} } \nonumber\\
&=  \Theta(n).\nonumber
\end{align}
The last equality uses the fact that all the distribution parameters are constants. The expected aggregate return with optimal loads,  $\Expc\left[X^*(t)\right]$, equals to $r$ at time $t=\tau^*$. Thus, $\tau^*=\frac{r}{s}=\Theta(1)$ and $\ell_i^*(\tau^*)=\frac{\tau^*}{\lambda_i}=\frac{r}{s\lambda_i}=\Theta(1)$. 

\begin{proof}[Proof of Theorem \ref{theorem3}] 
With the aforementioned reparametrizations of ${\lambda_i}$, $s$ and $\tau^*$, the HCMM algorithm defined in Algorithm \ref{def:HCMM} is identically applicable to the Weibull model. Proof of the asymptotic optimality of HCMM under the Weibull distribution follows the similar steps as in the proof for the exponential case in Section \ref{secIII-C} (unless specifically justified, e.g. (\ref{eq:th5})). We avoid rewriting these steps for the purpose of readability of the paper, but we note that the concentration inequalities used to establish the proof of Theorem \ref{theorem1} can be applied to a wide class of distributions including the Weibull distribution.
 \end{proof}
 As an implication of Theorem \ref{theorem3}, the induced expected execution time by HCMM algorithm is asymptotically constant, that is $\Expc[T_{\mathsf{HCMM}}] = \Theta(1)$; which was also the case for shifted exponential distribution. To compare with the uncoded scenario, we start by the following lemma which characterizes the extreme value of a sequence of Weibull random variables.

\begin{lemma}\label{lemma:maxw}
Let $\{T_i\}_{i=1}^{\infty}$ be a sequence of i.i.d. $\mathcal{W}(\alpha,\mu)$ random variables and $T^*_n=\max_{i \in [n]} T_i$ denote the maximum of the first $n$ variables. Then, 
\begin{equation}
\Expc\left[T^*_n \right] \geq \Theta \left( (\log n)^{1/\alpha}\right).\nonumber
\end{equation}
\end{lemma}
\begin{proof}
Consider the sequence of maximums $\{T^*_i\}_{i=1}^{\infty}$. From Markov's inequality, we have  $\frac{\Expc[T^*_n]}{t_n} \geq \Prob[T^*_n \geq t_n]$, for any $t_n>0$ and $n\in \mathbb{N}$. Pick $t_n=\frac{1}{\mu}\big(\log n \big)^{1/\alpha}$. Therefore,
\begin{align}
 \frac{\Expc[T^*_n]}{\frac{1}{\mu}\big(\log n \big)^{1/\alpha}} &\geq \Prob \left[ T^*_n \geq \frac{1}{\mu}\big(\log n \big)^{1/\alpha} \right] \nonumber\\
 &= 1-\Prob\left[T^*_n < \frac{1}{\mu}\big(\log n \big)^{1/\alpha}\right] \nonumber\\
 &= 1 - \prod_{i=1}^{n} \Prob\left[T_i < \frac{1}{\mu}\big(\log n \big)^{1/\alpha}\right]\nonumber\\
 &= 1- \left( 1-e^{-\log n} \right)^n \nonumber\\
 &= 1- \left( 1-\frac{1}{n} \right)^n.\nonumber
\end{align}
Therefore, 
\begin{equation}
\lim_{n \to \infty}\frac{\Expc[T^*_n]}{\frac{1}{\mu}\big(\log n \big)^{1/\alpha}} \geq \lim_{n \to \infty}1- \left( 1-\frac{1}{n} \right)^n=1-\frac{1}{e} > 0.63,\nonumber
\end{equation}
which implies $\Expc[T^*_n] \geq \Theta \left( (\log n)^{1/\alpha}\right)$.
\end{proof}

Now we complete the proof of Theorem \ref{theorem4}.

\begin{proof}[Proof of Theorem \ref{theorem4}]
Recall that $T_{\mathsf{UC}}$ denotes the completion time of the optimum uncoded distributed matrix multiplication algorithm across $n$ workers parametrized by tuples $\{(a_i, \mu_i, \alpha_i,)\}_{i=1}^{n}$. To bound the mean of $T_{\mathsf{UC}}$, assume that every machine is replaced with a stochastically faster machine with parameters $(\tilde{a},\tilde{\mu},\tilde{\alpha})$ where $\tilde{a} = \min_i a_i$, $\tilde{\mu} = \max_i \mu_i$ and $\tilde{\alpha} = \max_i \alpha_i$, i.e., the expected run-time of the latter scenario is no greater than that of the former one.  For the new set of $n$ identical machines, the optimal loading is uniform, i.e., $\widetilde{\ell}^*_i=\frac{r}{n}$. Let $\{\widetilde{T}_i\}_{i=1}^{n}$ denote the i.i.d. shifted Weibull run times for new set of machines which have CDFs of the form
 \begin{align}
\Prob\left[\widetilde{T}_i\leq t \right]=&1-e^{{-\left(\frac{\tilde{\mu}}{\widetilde{\ell}^*_i}\left(t-\tilde{a} \widetilde{\ell}^*_i\right)\right)^{\tilde{\alpha}}}} \nonumber 
\\=&1-e^{-\left(\tilde{\mu}\frac{n}{r}\left(t-\tilde{a} \frac{r}{n}\right)\right)^{\tilde{\alpha}}},\nonumber
\end{align}
for $t\geq  \frac{\tilde{a} r}{n}$. The mean of computation time for the new set of machines is 
\begin{equation}
\Expc\left[\widetilde{T}_{\mathsf{UC}} \right]= \Expc \left[\max_{i\in [n]} \widetilde{T}_i \right]=\frac{\tilde{a} r}{n}+ \Expc \left[\max_{i\in [n]} \doublewidetilde{T}_i \right],\nonumber
\end{equation}
where $\doublewidetilde{T}_i = \widetilde{T}_i - \frac{\tilde{a} r}{n}$ are i.i.d.  $\mathcal{W}(\tilde{\alpha},\tilde{\mu}\frac{n}{r})$ for all workers $i \in [n]$. Using Lemma \ref{lemma:maxw}, we can write 
\begin{equation}
\label{lowerbW}
\Expc[T_{\mathsf{UC}}] \geq \Expc \left[\widetilde{T}_{\mathsf{UC}} \right] \geq \frac{\tilde{a} r}{n}+ \Theta\left( (\log n)^{1/\tilde{\alpha}}\right)=\Theta\left( (\log n)^{1/\tilde{\alpha}}\right).\nonumber
\end{equation}
Comparing the best uncoded scheme with the proposed coded algorithm demonstrates that HCMM outperforms the best uncoded scheme by a factor of at least $\Theta\left( (\log n)^{1/\tilde{\alpha}}\right)$, i.e.,
\begin{equation}
\frac{\Expc[T_{\mathsf{UC}}]}{\Expc[T_{\mathsf{HCMM}}] } \geq \Theta \left( (\log n)^{1/\tilde{\alpha}}\right).\nonumber
\end{equation}
\end{proof}

%% file: 5-experiments.tex
In this section, we present our results both from simulations as well as from experiments over Amazon EC2 clusters. These results demonstrate how HCMM can provide significant speedups in comparison to state-of-the-art load allocation schemes.

\subsection{Numerical Analysis}
\label{subsec:numerical}
We now present numerical results evaluating the performance of HCMM. We consider both the shifted exponential model in (\ref{eq:1}) and the shifted Weibull model  in (\ref{eq:Wdist}) for run-time distributions in our simulations, assuming the unit seconds per row ($\text{s}/\text{row}$) for $a$ and $1/\mu$. The underlying computation task is to compute $r=10000$ inner products using a heterogeneous cluster of $n=100$ workers, where different scenarios for heterogeneity are considered. For each scenario under consideration, we implement the following load allocation schemes\footnote{For each scheme, the load number for each worker is approximated to the nearest larger integer using the $\mathtt{ceil()}$ function. For the practical large load regime considered in simulations, this rounding step has negligible impact on load allocation and on the overall results.}:

\begin{enumerate}
\item \textbf{Uniform Uncoded}: Each worker  is assigned an equal number of rows, i.e., $\ell_i=r/n$ for all workers $i$.

\item \textbf{Load-balanced Uncoded}: Each worker is assigned a load which is inversely proportional to its expected time for computing one inner product, i.e., for the shifted exponential model, $\ell_i\propto \mu_i/(a_i\mu_i+1)$, while for the shifted Weibull model, $\ell_i\propto \mu_i/(a_i\mu_i+\Gamma(1+1/\alpha_i))$ for all workers $i$. Furthermore, we set $\sum_{i=1}^n\ell_i=r$.

\item \textbf{Uniform Coded}: Equal number of coded rows are assigned to each worker. Redundancy is numerically optimized for minimizing the average computation time for receiving results of at least $r$ inner products at the master node. 

\item \textbf{HCMM}: Each worker is assigned the asymptotically optimal load allocation derived in Section \ref{secIII-B}, i.e., $\ell_i=\tau^*/\lambda_i$ for each worker $i$ according to (\ref{eq:lambda}) and (\ref{lopt}).


\end{enumerate}

\begin{figure}[t]
\includegraphics[width=0.5\textwidth]{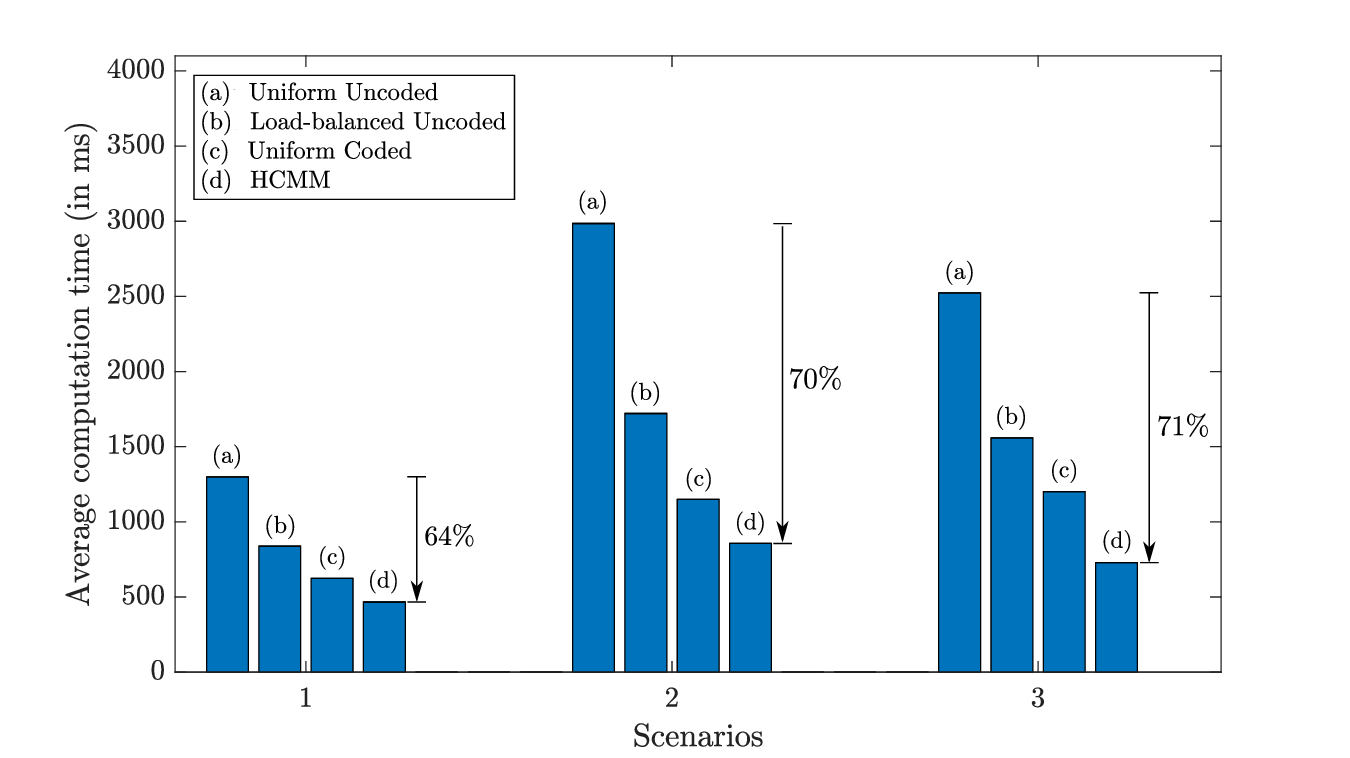}
\centering
\caption{Illustration of the performance gain of HCMM over the three benchmark schemes for the exponential run-time model. Among the three scenarios, HCMM achieves a performance improvement of up to $71\%$ over Uniform Uncoded, up to $53\%$ over Load-balanced Uncoded, and up to $39\%$ over Uniform Coded. Furthermore, the coding redundancy ${\sum_{i=1}^n\ell_i}/{r}$ for the three scenarios is in the range of $1.41-1.46$ for HCMM and in the range of $2.3-2.8$ for Uniform Coded. This demonstrates the efficient utilization of resources by HCMM.}
\label{fig:num_expo}
\end{figure}

For simulations under the shifted exponential model, we consider the following three scenarios:
\begin{itemize}
\item \textbf{Scenario $\mathbf{1}$ ($\mathbf{2}$-mode heterogeneity)}: $(a_i,\mu_i)=(1,1)$ for $50$ workers, and $(a_i,\mu_i)=(4,0.5)$ for the other $50$ workers.
\item \textbf{Scenario $\mathbf{2}$ ($\mathbf{3}$-mode heterogeneity)}: $(a_i,\mu_i)=(1,0.5)$ for $25$ workers, $(a_i,\mu_i)=(4,2)$ for $25$ workers, and $(a_i,\mu_i)=(12,0.25)$ for the remaining $50$ workers.
\item \textbf{Scenario $\mathbf{3}$ (Random heterogeneity)}: For each worker $i$, parameters $a_i$ and $\mu_i$ are sampled from the sets $\{1,4,12\}$, $\{0.5,2,0.25\}$, respectively and all uniformly at random. 
\end{itemize}

\begin{figure}[t]
\includegraphics[width=0.5\textwidth]{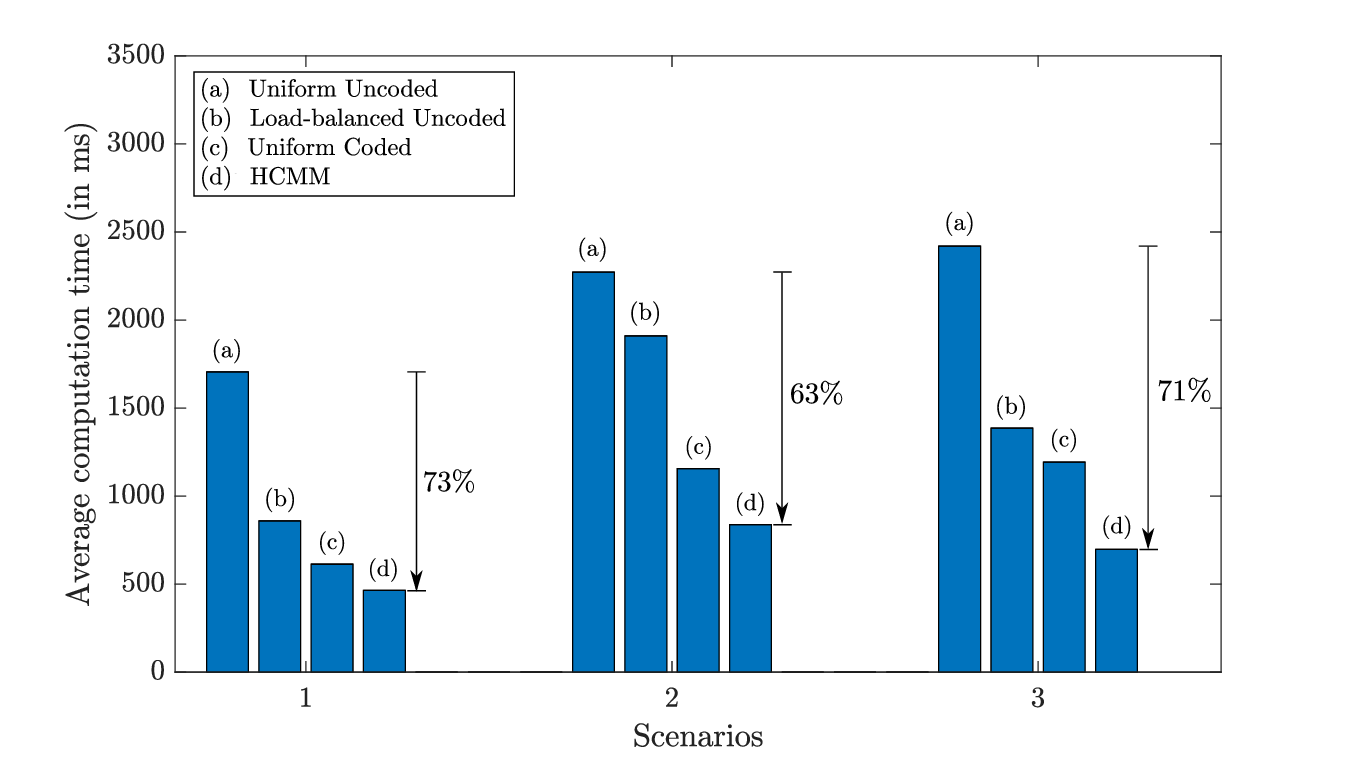}
\centering
\caption{Illustration of the performance gain of HCMM over the three benchmark schemes for Weibull model for run-time. Among the three scenarios, HCMM achieves a performance improvement of up to $73\%$ over Uniform Uncoded, up to $56\%$ over Load-balanced Uncoded, and up to $42\%$ over Uniform Coded. Furthermore, the coding redundancy $\sum_{i=1}^n\ell_i/r$ for the three scenarios is in the range of $1.30-1.42$ for HCMM and in the range of $2.0-2.5$ for Uniform Coded. This demonstrates the efficient utilization of resources by HCMM.}
\label{fig:num_weibl}
\end{figure}

The following three scenarios are considered for simulations under the shifted Weibull distribution for run-times:
\begin{itemize}
\item \textbf{Scenario $\mathbf{1}$ ($\mathbf{2}$-mode heterogeneity)}: $(a_i,\mu_i,\alpha_i)=(1,1,1.2)$ for $50$ workers, and $(a_i,\mu_i,\alpha_i)=(4,0.5,0.8)$ for the other $50$ workers.
\item \textbf{Scenario $\mathbf{2}$ ($\mathbf{3}$-mode heterogeneity)}: $(a_i,\mu_i,\alpha_i)=(1,0.5,0.9)$ for $25$ workers, $(a_i,\mu_i,\alpha_i)=(4,2,1.2)$ for $25$ workers, and $(a_i,\mu_i,\alpha_i)=(12,0.25,1.5)$ for the remaining $50$ workers.
\item \textbf{Scenario $\mathbf{3}$ (Random heterogeneity)}: For each worker $i$, parameters $a_i$, $\mu_i$ and $\alpha_i$ are sampled from the sets $\{1,4,12\}$, $\{0.5,2,0.25\}$ and $\{0.9,1.2,1.5\}$, respectively and all uniformly at random. 
\end{itemize}

Fig. \ref{fig:num_expo} and \ref{fig:num_weibl} illustrate the performance comparison of the four schemes for the two run-time models. We make the following conclusions from the results.

\begin{itemize}
    \item HCMM significantly outperforms the benchmark load allocation schemes. In particular, for the shifted exponential model, HCMM provides speedups of up to $71\%$ over Uniform Uncoded, up to $53\%$ over Load-balanced Uncoded, and up to $39\%$ over Uniform Coded, among the three scenarios. When the machine run-time is assumed to have a shifted Weibull distribution, among the three scenarios HCMM results in gains of up to $73\%$, $56\%$ and $42\%$ over Uniform Uncoded, Load-balanced Uncoded, and Uniform Coded respectively.
    
    \item The coding redundancy $\sum_{i=1}^n\ell_i/r$ for Uniform Coded is higher in comparison to the one for HCMM. In particular, for simulations under the shifted exponential model, the coding redundancy for the three scenarios is in the range of $2.3-2.8$ for Uniform Coded and in the range of $1.41-1.46$ for HCMM. For simulations under the shifted Weibull distribution, the coding redundancy is in the range of $2.0-2.5$ for Uniform Coded, while for HCMM, it is in the range $1.30-1.42$. This demonstrates that HCMM leads to a better utilization of computing resources. 
    
    \item Both Load-balanced Uncoded and Uniform Coded improve upon the performance of Uniform Uncoded. In Load-balanced Uncoded scheme, assigning larger loads to  faster machines leads to better performance, while for Uniform Coded, repeated computations lead to better performance as the master does not need to wait for all the results. HCMM provides the best expected execution time among the four schemes as it combines the gains of Load-balanced Uncoded and Uniform Coded by employing efficient load balancing along with minimal number of redundant computations.
\end{itemize}

Next, we present the results from our experiments over Amazon EC2 clusters. These results show agreement with our numerical studies.

\subsection{Experiments using Amazon EC2 machines}
\label{subsec:aws}
We use Python with \texttt{mpi4py} package \cite{dalcin2005mpi} to implement our developed HCMM scheme over Amazon EC2 clusters. To emulate the straggler effects in large-scale systems \cite{dean2013tail}, we inject artificial delays.\footnote{Artificial delays are injected since stragglers are rarely observed in small clusters in Amazon EC2. Though other emerging platforms such as federated learning, computation with deadline, mobile edge computing, fog computing, etc., still suffer from stragglers where our ideas can be employed \cite{li2017coding}.} This is achieved by selecting some workers to be stragglers at the beginning of experiments and slowing down each such worker by making it wait for $3$ times the amount of time it spends in computation before it sends its results to the master. This is done using the \texttt{sleep()} function in \texttt{time} package. For each scenario, the choice of stragglers is made by drawing a sample from the Bernoulli$(0.5)$ distribution for each worker, i.e., each worker is chosen to be a straggler with probability $0.5$. 

In line with our simulation studies, we compare the performance of HCMM with the three benchmark load allocation schemes. For Load-balanced Uncoded, the number of uncoded rows $\ell_i$ assigned to a worker $i$ is proportional to the number of virtual CPUs, and the loads are normalized to have a sum equal to $r$. For the  encoding and the decoding steps for Uniform Coded as well as HCMM, we utilize the Luby transform (LT) codes with peeling decoder which provides nearly linear decoding complexity \cite{mackay2003information}. Utilization of LT codes for distributed computing is proposed in \cite{mallick2018rateless} as well. However, they perform a homogeneous load allocation by assigning an equal number of rows of the encoded data matrix to each worker and hence do not capture the heterogeneity of the computing cluster in distributing the encoded data matrix. Towards this end, we relax our goal of recovering all the inner products from any $r$ of the coded inner products to recovering all the inner products from any $r'=r(1+\epsilon)$ coded inner products with high probability. Ideally, we would like to have $\epsilon > 0$ to be as small as possible. In our experiments, we keep $r=10000$, and based on the results in \cite{mallick2018rateless}, we use the robust Soliton degree distribution with $(c,\delta)=(0.03,0.1)$ and select $\epsilon=0.13$, where $c$ is a tuning parameter and $\delta$ is a bound on the probability of failure of decoding from a certain number of received coded inner products (see \cite{mallick2018rateless} for details). Therefore, for both HCMM and Uniform Coded, we design the load allocation such that the master needs to wait only for $r'=11300$ coded inner products. The total computation time is equal to the waiting time for $r'=11300$ results plus the average time for decoding the $r=10000$ inner products from the received $r'=11300$ coded inner products.\footnote{The average time for decoding $r=10000$ inner products from any $r(1+\epsilon)$ coded inner products is obtained using a \textBF{m4.xlarge} instance.} For HCMM,  we use the shifted exponential distribution for estimating the computation model for each worker.

\begin{figure}[t]
\includegraphics[width=0.5\textwidth]{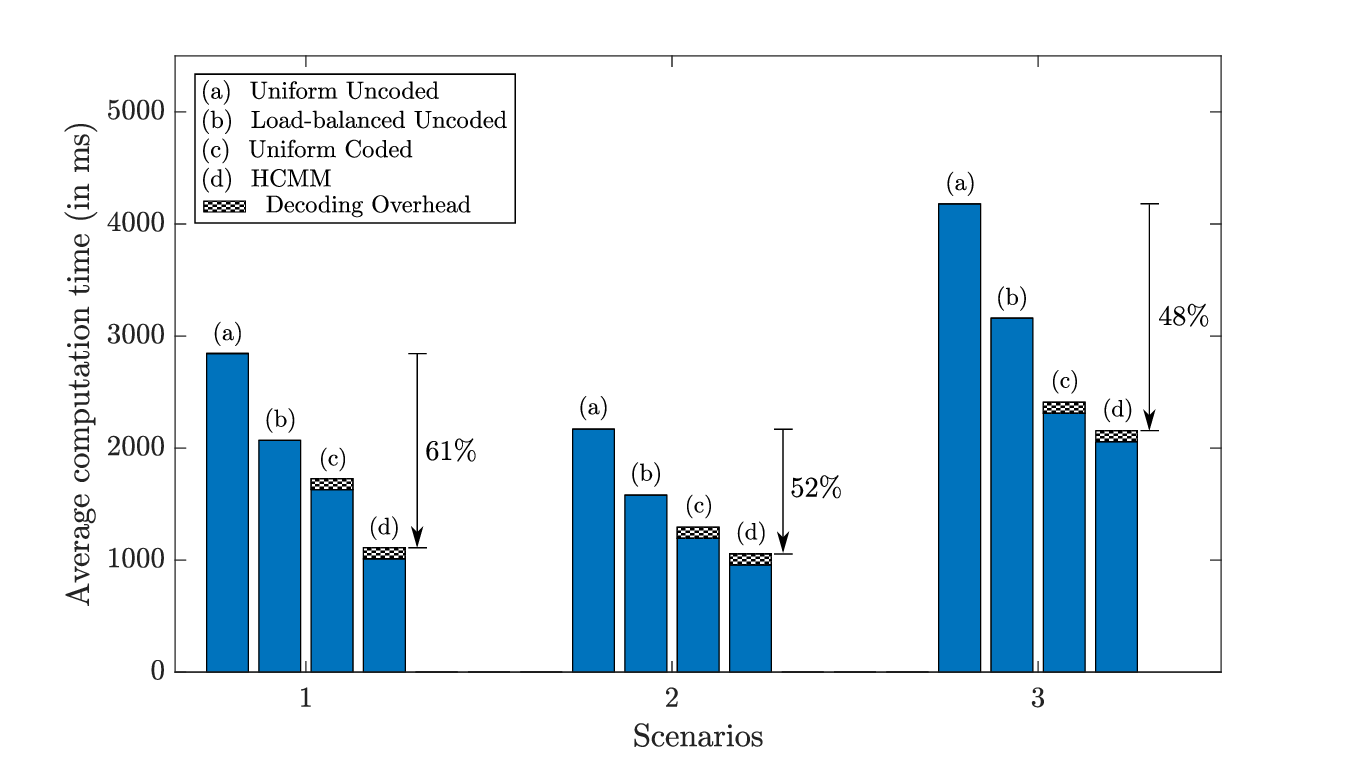}
\centering
\caption{Illustration of the performance gain of HCMM over the three benchmark schemes. Among the three scenarios, HCMM achieves a performance improvement of up to $61\%$ over Uniform Uncoded, up to $46\%$ over Load-balanced Uncoded, and up to $36\%$ over Uniform Coded. Furthermore, the coding redundancy $\sum_{i=1}^n\ell_i/r$ for the three scenarios is approximately $1.4$ for HCMM and in the range of $2.12-2.26$ for Uniform Coded. Therefore, HCMM gives the best overall execution time among the four scenarios with minimal coding overhead.}
\label{fig:amazon}
\end{figure}

For performance comparison of the four schemes, we consider the following three computing scenarios:
\begin{itemize}
\item \textbf{Scenario $\mathbf{1}$}: Each row has $500000$ elements. We use a heterogeneous cluster of $11$ machines -- one master of instance type \textBF{m4.xlarge}, four workers of instance type \textBF{r4.2xlarge}, and six workers of instance type \textBF{r4.xlarge}.
\item \textbf{Scenario $\mathbf{2}$}: Each row has $500000$ elements. We use a heterogeneous cluster of $16$ machines -- one master of instance type \textBF{m4.xlarge}, six workers of instance type \textBF{r4.2xlarge}, and nine workers of instance type \textBF{r4.xlarge}.
\item \textbf{Scenario $\mathbf{3}$}: Each row has $1000000$ elements. We use the same heterogeneous cluster as in the previous scenario.
\end{itemize}

Fig. \ref{fig:amazon} provides a performance comparison of HCMM with the benchmark load allocation schemes for the three scenarios, where the decoding time is taken into account as well. Fig. \ref{fig:instances} presents the typical cumulative distribution functions for the instances used in the experiments. We make the following conclusions from the results:
\begin{itemize}
\item As demonstrated in Fig. \ref{fig:instances}, the shifted exponential model is a good first order fit for the run-times of the workers.  
\item HCMM achieves significant speedups over the benchmark load allocation policies. In particular, HCMM combined with LT codes provides gains in the overall execution time of up to $61\%$ over Uniform Uncoded, up to $46\%$ over Load-balanced Uncoded, and up to $36\%$ over Uniform Coded. 
\item As presented in Table \ref{tab:HCMM_red_ws}, HCMM has significantly lower total computation load compared to Uniform Coded. Hence, HCMM leads to efficient utilization of the computing resources, combining the benefits of both Load-balanced Uncoded and Uniform Coded schemes.
\end{itemize}

\begin{figure}
\begin{subfigure}{\linewidth}
\centering
\includegraphics[width=0.8\textwidth]{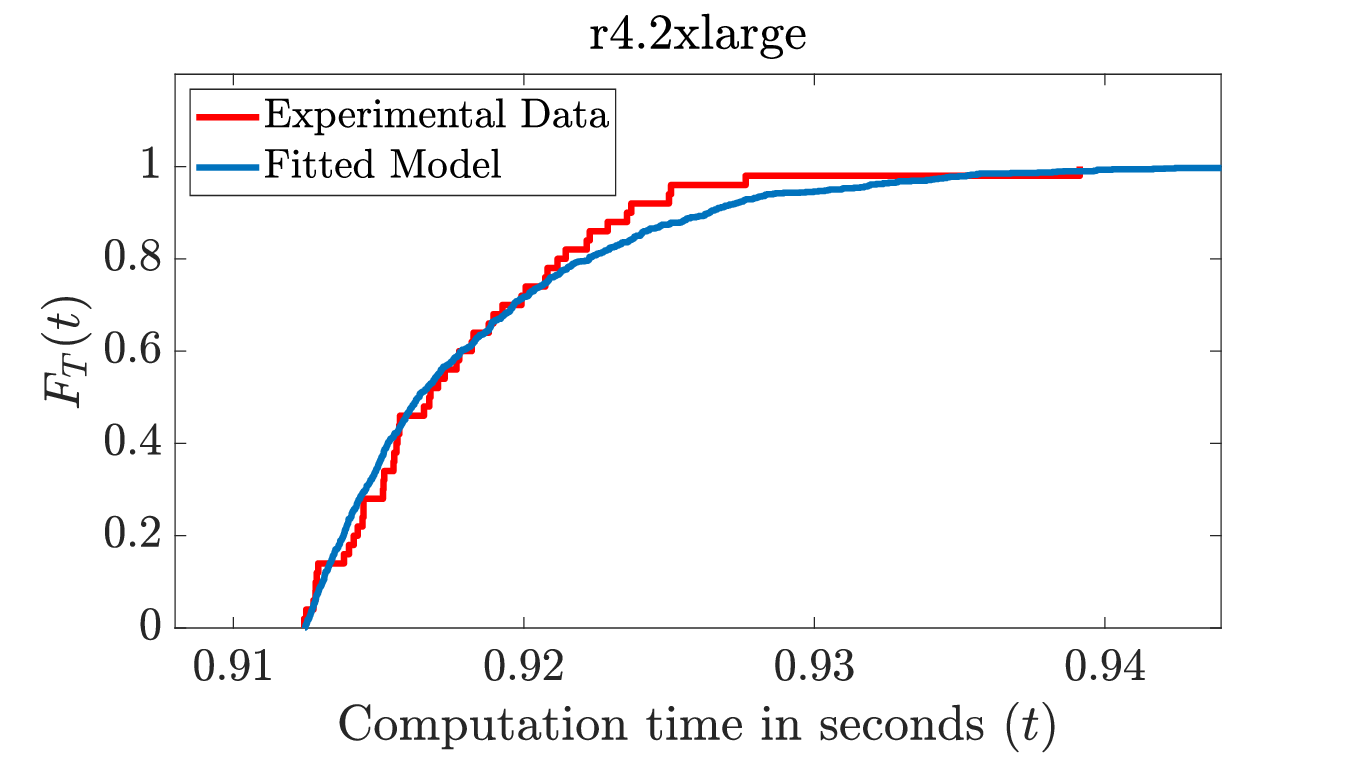}
\caption{$(a,1/\mu)=(1.37\times 10^{-3} \text{s}/\text{row}, 8.25\times 10^{-6} \text{s}/\text{row})$}
\vspace{.4cm}
\label{fig:r42xlarge}
\end{subfigure}%
\\
\begin{subfigure}{\linewidth}
\centering
\includegraphics[width=0.8\textwidth]{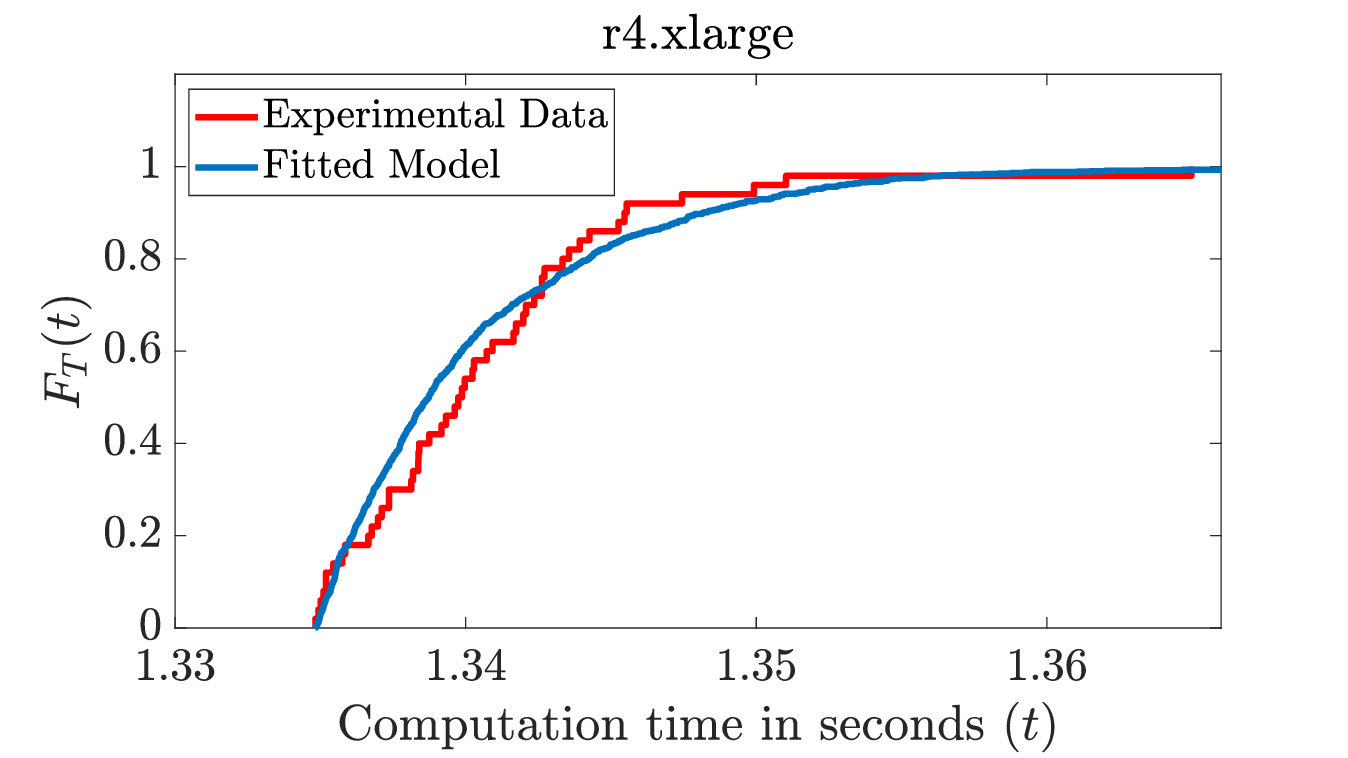}
\caption{$(a,1/\mu)=(2.00\times 10^{-3} \text{s}/\text{row}, 8.72\times 10^{-6} \text{s}/\text{row})$}
  \vspace{.5cm}
  \label{fig:r4xlarge}
\end{subfigure}
\caption{Typical empirical cumulative distribution functions for two instances used in Scenario 3 of our experiments. The measurements were taken in the absence of any manual delay. As demonstrated here, shifted exponential distribution is a good model for the task execution  time in EC2 machines.}
 \label{fig:instances}
\end{figure}

\begin{table}[t]
\begin{center}
\caption{Total computation load ($\sum_{i=1}^n\ell_i$) of HCMM and Uniform Coded}
\label{tab:HCMM_red_ws}
\begin{tabular}{|c|c|c|c|c|}
\hline
Scenario & $n$  &  HCMM& Uniform Coded	 \\ \hline
$1$ & $10$  & $11397$& $22600$	 \\ \hline
$2$ & $15$   & $11402$	 &  $21201$\\ \hline
$3$ & $15$ & $11403$	   &  $21201$\\ \hline

\end{tabular}
\end{center}
\end{table}

These results demonstrate that HCMM can provide significant speedups in large-scale computing environments.

%% file: 6-gen_budget_constr.tex
In this section, we consider the optimization problem in (\ref{eq:2}) under the shifted exponential distribution with a monetary constraint for carrying out the overall computation. Running computation tasks on a commodity server costs depending on several factors including CPU, memory, ECU, storage, bandwidth, etc. Different cloud computing platforms employ different pricing policies, and these need to be taken into account for developing efficient task allocation and execution algorithms \cite{su2013cost,deelman2008cost,kondo2009cost,yi2012monetary,malawski2015algorithms}. For example, Table \ref{amazon} summarizes the cost per hour of using Amazon EC2 clusters with different parameters (at the time of writing this manuscript) \cite{amazon-price}. In this section, we take into account the monetary constraint in the optimization problem in (\ref{eq:2}) and provide a heuristic algorithm towards finding the optimal load allocation under cost budget constraint. 
 
\begin{table}[t]
\begin{center}
\caption{Amazon EC2 Pricing for Linux}
\label{amazon}
\begin{tabular}{|c|c|c|c|c|c|}
\hline
machine	& vCPU	 & ECU	 & \begin{tabular}{@{}c@{}}Memory    \\ (GiB) \end{tabular}	 & \begin{tabular}{@{}c@{}@{}}Instance \\ Storage    \\ (GB) \end{tabular}	 &  \begin{tabular}{@{}c@{}}price    \\ (/Hour) \end{tabular}   \\ \hline
 \textBF{m3.medium}	& 1 & 3 & 3.75 &  1$\times$4 SSD&   \$0.077  \\ \hline
 \textBF{m3.large}	& 2 & 6.5 &7.5  & 1$\times$32 SSD &  \$0.154   \\ \hline
\textBF{m3.xlarge}	& 4 & 13 &  15&   2$\times$40 SSD &\$0.308  \\ \hline
 \textBF{m3.2xlarge}	& 8 & 26 &  30 &  2$\times$80 SSD &\$0.616  \\ \hline
\end{tabular}
\end{center}
\end{table}

 We now present the precise problem formulation we are interested in. For a computation task and a given set of $N$ machines, the goal is to minimize the expected run-time while satisfying the budget constraint $C$, that is
\begin{equation}
\label{eq:opt-cnstrnd}
\begin{aligned}
 \mathcal{P}_{\textnormal{main-constrained}}:\,\,\,\,\,& \underset{\mathbf{\ell}}{\text{minimize}}
&& \Expc[T_\mathsf{CMP}] \\
&  \text{subject to}
&& \sum_{i=1}^{N}c_i \mathbbm{1}_{\{\ell_i>0\}}\Expc[T_\mathsf{CMP}]\leq C,
\end{aligned}
\end{equation}
where $c_i$ represents the cost per time unit of using machine $i\in[N]$.
According to the pricing polices provided by AWS, e.g. Table \ref{amazon}, a linear model for cost (versus performance parameters) is intuitive and convincing. Considering the last two rows of Table II for instance, doubling the parameters results in doubled cost. To be general, we model the computation cost of a single machine as $c=\kappa \mu^\gamma$ per unit of time, which captures a convex dependency of the speed parameter $\mu$ for constants $ \gamma\geq 1$ and $\kappa>0$. 

 We assume that there are $K$ types of machines parameterized with $\{(a_k,\mu_k)\}_{k=1}^{K}$, and $N_k,k\in[K]$ of each type is available to run a distributed computation task, where $N=\sum_{k=1}^{K} N_k$ is the total number of available machines. We also assume that $\mu_1 \leq \cdots \leq \mu_K$ and $a_1 \mu_1 = \cdots = a_K \mu_K = \xi$ for a constant $\xi$.\footnote{The latter assumption can be intuitively justified as follows. If a machine is $c$ times more powerful than another machine, as the first order estimation, one can assume that both the shift ($a_k$) and the straggling parameter ($\mu_k$) of the computation are $c$ times stronger.} As we showed in Theorem \ref{theorem1}, HCMM is asymptotically optimal (i.e. optimal within a vanishing deviation) regarding the average run-time. In this section, we also consider the asymptotic regime, i.e. for large enough number of machines and hence HCMM attains the optimality per $\mathcal{P}_{\textnormal{main}}$ in (\ref{eq:2}).
 
 
 The following lemma states a useful observation regarding the solutions to the constrained problem $\mathcal{P}_{\textnormal{main-constrained}}$ and the  minimum possible cost for carrying out a computation task. 

 \begin{lemma} 
 \label{max-min}
 HCMM is the (asymptotic) solution to the feasible  $\mathcal{P}_{\textnormal{main-constrained}}$. Moreover, given a computation task and a set of machines, decreasing the number of fastest (slowest) machines in HCMM, results in smaller (greater) expected cost. And, the minimum (maximum) cost of HCMM is induced by running the task only on any number of the slowest (fastest) machines. 
 \end{lemma}
 \begin{proof}
 We first argue that if the budget-constrained problem defined in  $\mathcal{P}_{\textnormal{main-constrained}}$ is feasible, then HCMM determines the asymptotically optimal load allocation. Consider a set of $N$ machines and assume that $M$ of them are assigned non-zero loads in the optimal budget-constrained scheme. Now, one can run HCMM load allocation over the set of these $M$ machines and according to asymptotic optimality results, HCMM asymptotically attains the optimal run-time  while satisfying the budget constraint.

 Now assume that $n_k$ number of type $k \in [K]$ machine is used. Then, by assigning the loads obtained from HCMM and the result of Theorem \ref{theorem1}, the induced expected cost (for large number of machines) can be written as 
 \begin{align} \label{cost}
 \text{cost}\big(\text{HCMM}(n_1,\cdots,n_K)\big) &=\tau^* \sum_{k=1}^{K} n_k c_k \nonumber\\
 & = \frac{r}{s} \sum_{k=1}^{K} n_k c_k \nonumber\\
 & = \frac{r}{\sum_{k=1}^{K} \frac{n_k \mu_k}{1 + \mu_k \lambda_k}} \sum_{k=1}^{K} n_k \kappa \mu^\gamma_k \nonumber\\
 & = \kappa r x_{\xi} \frac{\sum_{k=1}^{K} n_k\mu_k^\gamma}{\sum_{k=1}^{K} n_k\mu_k },
 \end{align}
 where $x_{\xi} = 1 + \mu_k \lambda_k$ is the solution to the equation $e^{x_{\xi} - \xi - 1} = x_{\xi}$ for all machine type $k \in [K]$. In another scenario, assume that we remove one machine of type $K$ (the fastest machine type) and run HCMM accordingly, i.e. $n_k$ of type $k \in [K-1]$ and $n_K - 1$ of type $K$. The expected cost of this scenario can be written as follows:
 \begin{align}\label{eq:n_K - 1}
     \text{cost}\big(\text{HCMM} & (n_1,\cdots,n_K-1)\big) \nonumber\\
    & = \kappa r x_{\xi} \frac{\sum_{k=1}^{K-1} n_k\mu_k^\gamma + (n_K - 1) \mu^\gamma_K}{\sum_{k=1}^{K-1} n_k\mu_k  + (n_K - 1) \mu_K } \nonumber \\
     & \overset{(f)}{\leq} \kappa r x_{\xi} \frac{\sum_{k=1}^{K} n_k\mu_k^\gamma}{\sum_{k=1}^{K} n_k\mu_k } \nonumber\\
     &= \text{cost}\big(\text{HCMM}(n_1,\cdots,n_K)\big),
 \end{align}
 where inequality $(f)$ can be easily verified given that $\mu_1 \leq \cdots \leq \mu_K$.
 We can iteratively apply the same argument and conclude that the minimum expected cost is achieved when only the slowest machines are used, that is
 \begin{equation}\label{eq:Cm}
C_{\text{min}} \coloneqq \text{cost}\big(\text{HCMM}(n_1,0,\cdots,0)\big) = \kappa r x_{\xi} \mu^{\gamma-1}_1,
 \end{equation}
for any $1 \leq n_1 \leq N_1$. Similar to (\ref{eq:n_K - 1}), one can show that reducing the number of participating slowest machines increases the induced expected cost of HCMM, that is
\begin{align}\label{eq:n_1 - 1}
     \text{cost}\big(\text{HCMM}(& n_1 - 1,\cdots, n_K)\big) \nonumber\\
     & \geq \text{cost}\big(\text{HCMM}(n_1,\cdots,n_K)\big).
\end{align}
 Therefore, applying (\ref{eq:n_1 - 1}) iteratively shows that the maximum expected cost occurs when only the fastest machines are employed, that is
\begin{equation}\label{eq:CM}
C_{\text{max}} \coloneqq \text{cost}\big(\text{HCMM}(0,\cdots,0,n_K)\big) = \kappa r x_{\xi} \mu^{\gamma-1}_K,\nonumber
 \end{equation}
 for any $1 \leq n_K \leq N_K$.
 \end{proof}

  Lemma \ref{max-min} implies that if the available budget $C$ is less than $C_{\text{min}}$ defined in (\ref{eq:Cm}), then $\mathcal{P}_{\textnormal{main-constrained}}$ is infeasible and it is impossible to run the task on the given set of machines while satisfying the budget constraint. Moreover, reducing one machine from the available set of fastest machines along with HCMM results in a lower expected cost; and reducing the number of participating slowest machines results in a larger expected cost. 
 
Now that HCMM asymptotically solves the feasible budget-constrained problem in (\ref{eq:opt-cnstrnd}), i.e. for $C \geq C_{\text{min}}$, finding the optimal number of machines of each type to use in HCMM requires combinatorial search over all possible allocations.  However, as Lemma \ref{max-min} suggests, using faster machines induces a larger cost. Further, the computation time increases if we decrease the number of machines. This is the motivation behind our heuristic algorithm for an efficient search to find the number of machines of each type to include in HCMM, which we describe next.  
 \begin{algorithm}[h!]
\caption{Heuristic Search}\label{alg}
\begin{algorithmic}[1]
\Procedure{Heuristic Search}{}
\State $(n_1,\cdots,n_K) \gets (N_1,\cdots,N_K)$
\BState \emph{top}:
\State Run HCMM with $(n_1,\cdots,n_K)$
\If {$\text{cost}\big(\text{HCMM}(n_1,\cdots,n_K)\big) > C$} 
\State $n_j \gets n_j-1$ where $j= \max \{ k: n_k>0 \} $
\State  \textbf{goto} \emph{top}
\Else
\State \Return $(n_1,\cdots,n_K)$
\EndIf
\EndProcedure
\end{algorithmic}
\end{algorithm}

 First, Algorithm \ref{alg} runs HCMM algorithm using all machines, i.e., $n_k=N_k$ for each $k\in [K]$. Then, it calculates the corresponding cost according to (\ref{cost}). If the cost is $>C$, it starts to decrease the number of available fastest machines, i.e. $n_K\leftarrow n_K-1$, and runs HCMM again. While the cost is $>C$, the algorithm keeps decreasing the number of used fast machines till $n_K=0$. Then, the algorithm sets $n_K=0$ and starts decreasing $n_{K-1}$ and so on, until a feasible cost is achieved. Thus, the algorithm returns $(N_1,\cdots,N_{j},n_{j+1},0,\cdots,0)$ which is the first tuple that satisfies the cost constraint. Therefore, the search space complexity of the heuristic is $\mathcal{O}(N_1+\cdots+N_K)=\mathcal{O}(N)$ which is more efficient than the exhaustive search where the complexity is $\mathcal{O}(N_1\cdots N_K)$. The pseudo-code in Algorithm \ref{alg} summarizes the heuristic.

\begin{example*}\label{ex}
In this example, we consider two different scenarios to  demonstrate the application of the proposed heuristic search algorithm. For the cost model, we assume $\gamma=2$ and $\kappa=1$, i.e. $c=\mu^2$. Further, we consider the task of computing $r=100$ equations.

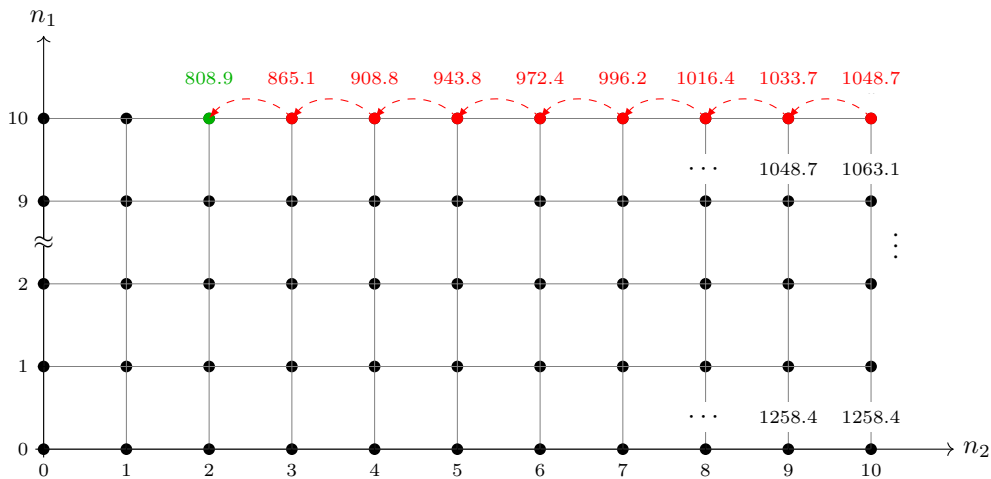
\begin{figure*}[t]
\centering
\begin{tikzpicture}
 \foreach \x in {0,...,10}
    \foreach \y in {0,...,4}
    {
    \fill (\x,\y) circle (2pt);
    }
\draw[help lines] (0,0) grid (10,4);
\draw[->] (-0.1,0) -- (11,0) node[right] {$n_2$};
\draw[->] (0,-0.1) -- (0,5) node[above] {$n_1$};
\draw (0,2) -- node[fill=white,inner sep=-1.25pt,outer sep=0,anchor=center]{$\approx$} (0,3);
\foreach \x/\xtext in {0/0 , 1/1, 2/2, 3/3, 4/4 , 5/5,6/6,7/7,8/8,9/9,10/10}
\draw[shift={(\x,0)}] (0pt,2pt) -- (0pt,-2pt) node[below] {$\scriptstyle\xtext$};
\foreach \y/\ytext in {0/0 ,1/1, 2/2}
\draw[shift={(0,\y)}] (2pt,0pt) -- (-2pt,0pt) node[left] {$\scriptstyle\ytext$};
\foreach \y/\ytext in {9/9,10/10}
\draw[shift={(0,\y-6)}] (2pt,0pt) -- (-2pt,0pt) node[left] {$\scriptstyle\ytext$};
 \foreach \x in {3,...,10}
    \foreach \y in {4}
    {
    \fill[fill=red] (\x,\y) circle (2pt);
    }
    \fill[fill=black!30!green] (2,4) circle (2pt);
\draw (10cm,4.2cm) node[color=black,above,fill=white] {$\scriptstyle a$};
\draw (8cm,4.2cm) node[color=black,above,fill=white] {$ \cdots$};
\draw (10.3cm,2.2cm) node[color=black,above,fill=white] {$ \vdots$};

\draw (10cm,4.3cm) node[color=red,above,fill=white] {$\scriptstyle 1048.7$};
\draw (9cm,4.3cm) node[color=red,above,fill=white] {$\scriptstyle 1033.7$};
\draw (8cm,4.3cm) node[color=red,above,fill=white] {$\scriptstyle 1016.4$};
\draw (7cm,4.3cm) node[color=red,above,fill=white] {$\scriptstyle 996.2$};
\draw (6cm,4.3cm) node[color=red,above,fill=white] {$\scriptstyle 972.4$};
\draw (5cm,4.3cm) node[color=red,above,fill=white] {$\scriptstyle 943.8$};
\draw (4cm,4.3cm) node[color=red,above,fill=white] {$\scriptstyle 908.8$};
\draw (3cm,4.3cm) node[color=red,above,fill=white] {$\scriptstyle 865.1$};
\draw (2cm,4.3cm) node[color=black!30!green,above,fill=white] {$\scriptstyle 808.9$};

\draw (10cm,3.2cm) node[color=black,above,fill=white] {$\scriptstyle 1063.1$};

\draw (9cm,3.2cm) node[color=black,above,fill=white] {$\scriptstyle 1048.7$};

\draw (8cm,3.2cm) node[color=black,above,fill=white] {$\cdots$};

\draw (10cm,0.2cm) node[color=black,above,fill=white] {$\scriptstyle 1258.4$};

\draw (9cm,0.2cm) node[color=black,above,fill=white] {$\scriptstyle 1258.4$};

\draw (8cm,0.2cm) node[color=black,above,fill=white] {$\cdots$};

 \coordinate (y10) at (10cm,4cm);
  \coordinate (y9) at (9cm,4cm);
  \coordinate (y8) at (8cm,4cm);
  \coordinate (y7) at (7cm,4cm);
  \coordinate (y6) at (6cm,4cm);
  \coordinate (y5) at (5cm,4cm);
  \coordinate (y4) at (4cm,4cm);
  \coordinate (y3) at (3cm,4cm);
  \coordinate (y2) at (2cm,4cm);
 \draw[color=red,dashed,-latex] (y10) to [bend right=50] (y9);
 \draw[color=red,dashed,-latex] (y9) to [bend right=50] (y8);
 \draw[color=red,dashed,-latex] (y8) to [bend right=50] (y7);
 \draw[color=red,dashed,-latex] (y7) to [bend right=50] (y6);
 \draw[color=red,dashed,-latex] (y6) to [bend right=50] (y5);
 \draw[color=red,dashed,-latex] (y5) to [bend right=50] (y4);
 \draw[color=red,dashed,-latex] (y4) to [bend right=50] (y3);
 \draw[color=red,dashed,-latex] (y3) to [bend right=50] (y2);
\end{tikzpicture}
\caption{Total cost associated with every pair of $(n_1,n_2)$; $0\leq n_1,n_2 \leq 10$.}
\label{table-cost}
\end{figure*}


\begin{figure*}[t]
\centering
\begin{tikzpicture}
 \foreach \x in {0,...,10}
    \foreach \y in {0,...,4}
    {
    \fill (\x,\y) circle (2pt);
    }
\draw[help lines] (0,0) grid (10,4);
\draw[->] (-0.1,0) -- (11,0) node[right] {$n_2$};
\draw[->] (0,-0.1) -- (0,5) node[above] {$n_1$};
\draw (0,2) -- node[fill=white,inner sep=-1.25pt,outer sep=0,anchor=center]{$\approx$} (0,3);
\foreach \x/\xtext in {0/0 , 1/1, 2/2, 3/3, 4/4 , 5/5,6/6,7/7,8/8,9/9,10/10}
\draw[shift={(\x,0)}] (0pt,2pt) -- (0pt,-2pt) node[below] {$\scriptstyle\xtext$};
\foreach \y/\ytext in {0/0 ,1/1, 2/2}
\draw[shift={(0,\y)}] (2pt,0pt) -- (-2pt,0pt) node[left] {$\scriptstyle\ytext$};
\foreach \y/\ytext in {9/9,10/10}
\draw[shift={(0,\y-6)}] (2pt,0pt) -- (-2pt,0pt) node[left] {$\scriptstyle\ytext$};
 \foreach \x in {3,...,10}
    \foreach \y in {4}
    {
    \fill[fill=red] (\x,\y) circle (2pt);
    }
    \fill[fill=black!30!green] (2,4) circle (2pt);
\draw (10cm,4.2cm) node[color=black,above,fill=white] {$\scriptstyle a$};
\draw (8cm,4.2cm) node[color=black,above,fill=white] {$ \cdots$};
\draw (10.3cm,2.2cm) node[color=black,above,fill=white] {$ \vdots$};

\draw (10cm,4.3cm) node[color=red,above,fill=white] {$\scriptstyle 5.24$};
\draw (9cm,4.3cm) node[color=red,above,fill=white] {$\scriptstyle 5.61$};
\draw (8cm,4.3cm) node[color=red,above,fill=white] {$\scriptstyle 6.05$};
\draw (7cm,4.3cm) node[color=red,above,fill=white] {$\scriptstyle 6.55$};
\draw (6cm,4.3cm) node[color=red,above,fill=white] {$\scriptstyle 7.15$};
\draw (5cm,4.3cm) node[color=red,above,fill=white] {$\scriptstyle 7.86$};
\draw (4cm,4.3cm) node[color=red,above,fill=white] {$\scriptstyle 8.73$};
\draw (3cm,4.3cm) node[color=red,above,fill=white] {$\scriptstyle 9.83$};
\draw (2cm,4.3cm) node[color=black!30!green,above,fill=white] {$\scriptstyle 11.23$};

\draw (10cm,3.2cm) node[color=black,above,fill=white] {$\scriptstyle 5.42$};

\draw (9cm,3.2cm) node[color=black,above,fill=white] {$\scriptstyle 5.82$};

\draw (8cm,3.2cm) node[color=black,above,fill=white] {$\cdots$};

\draw (10cm,0.2cm) node[color=black,above,fill=white] {$\scriptstyle  7.86$};

\draw (9cm,0.2cm) node[color=black,above,fill=white] {$\scriptstyle 8.73$};

\draw (8cm,0.2cm) node[color=black,above,fill=white] {$\cdots$};

 \coordinate (y10) at (10cm,4cm);
  \coordinate (y9) at (9cm,4cm);
  \coordinate (y8) at (8cm,4cm);
  \coordinate (y7) at (7cm,4cm);
  \coordinate (y6) at (6cm,4cm);
  \coordinate (y5) at (5cm,4cm);
  \coordinate (y4) at (4cm,4cm);
  \coordinate (y3) at (3cm,4cm);
  \coordinate (y2) at (2cm,4cm);
 \draw[color=red,dashed,-latex] (y10) to [bend right=50] (y9);
 \draw[color=red,dashed,-latex] (y9) to [bend right=50] (y8);
 \draw[color=red,dashed,-latex] (y8) to [bend right=50] (y7);
 \draw[color=red,dashed,-latex] (y7) to [bend right=50] (y6);
 \draw[color=red,dashed,-latex] (y6) to [bend right=50] (y5);
 \draw[color=red,dashed,-latex] (y5) to [bend right=50] (y4);
 \draw[color=red,dashed,-latex] (y4) to [bend right=50] (y3);
 \draw[color=red,dashed,-latex] (y3) to [bend right=50] (y2);
\end{tikzpicture}
\caption{Expected time associated with every pair of $(n_1,n_2)$; $0\leq n_1,n_2 \leq 10$.}
\label{table-Et}
\end{figure*}
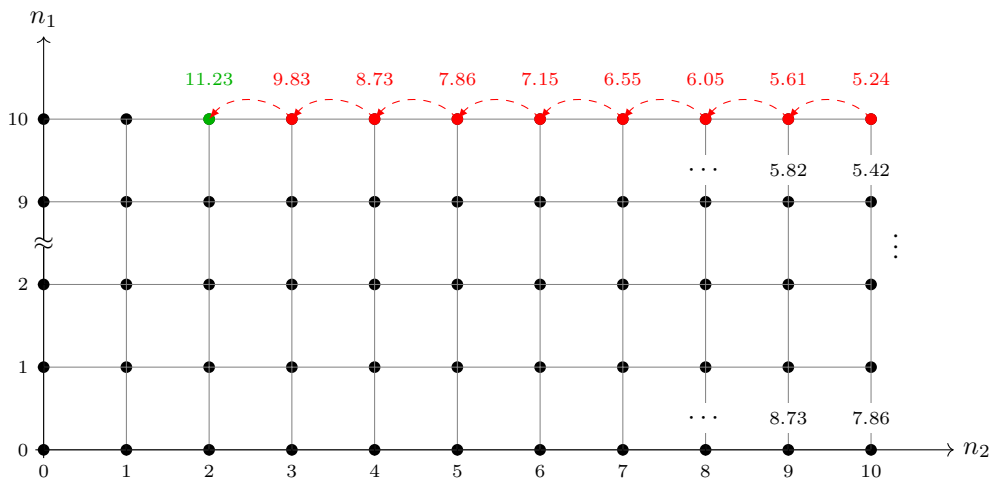

\begin{itemize}
\item \textbf{Scenario $\mathbf{1}$:} Two types of machines are available  parameterized by $(a_1,\mu_1)=(0.5,2)$ and $(a_2,\mu_2)=(0.25,4)$, assuming $10$ machines available of each type. Further, the available budget is $C=860$. Using Lemma \ref{max-min}, the minimum and maximum induced costs are $C_{\text{min}}=629.2$ and $C_{\text{max}}=1258.4$. As $C \geq C_{\text{min}}$, there exists an HCMM load allocation which is asymptotically optimal per (\ref{eq:opt-cnstrnd}). Applying the proposed heuristic search, it takes $9$ iterations (see Fig. \ref{table-cost} and \ref{table-Et}) to arrive at the load allocation $(n_1,n_2)=(10,2)$ which corresponds to the expected cost $808.9$ and average execution time $\Expc[T_{\mathsf{HCMM}}]=11.23$.

\item \textbf{Scenario $\mathbf{2}$:} Three types of machines are available which are parameterized by $(a_1,\mu_1)=(1,1),\, (a_2,\mu_2)=(0.5,2)$ and $(a_3,\mu_3)=(0.125,8)$, assuming $10$ machines of each type available. Further, the available budget is $C=475$. Using Lemma \ref{max-min}, the minimum and maximum induced costs for the task of computing $r=100$ equations are $C_{\text{min}}=314.6$ and $C_{\text{max}}=2516.8$ respectively. It takes $15$ iterations for the proposed heuristic search algorithm to arrive at the tuple $(n_1,n_2,n_3)=(10,6,0)$. This corresponds to the expected cost $486.2$ and the average time $\Expc[T_{\mathsf{HCMM}}]=14.3$.
\end{itemize}
\end{example*}

%% file: 8-conclusion.tex
In this paper, we proposed a coding framework for distributed matrix-vector multiplication in heterogeneous cloud computing environments. In particular, we considered two distributions for machines' run-times, i.e. shifted exponential and shifted Weibull  and tackled the intractable problem of minimizing the average run-time of a computation task over all possible load allocations by proposing a tractable alternative formulation. The solution to the alternative problem  established our proposed HCMM load allocation scheme which we proved to be asymptotically optimal. We also demonstrated the speedup of HCMM over three benchmark load allocation schemes and presented  both the numerical and the experimental results. Experiments over Amazon EC2 clusters  demonstrate that HCMM combined with LT codes and peeling decoders can provide significant gains in the average overall execution time. Moreover, we argued that HCMM is the asymptotically optimal allocation in budget-constrained scenarios as well, which led to providing a heuristic search in order to find a (sub)optimal load-machine assignment for a given set of machines while satisfying a pre-defined budget constraint.

%% file: 11-ack.tex
This material is based upon work supported by Defense Advanced Research Projects Agency (DARPA) under Contract No. HR001117C0053, ARO award W911NF1810400, NSF grants CCF-1703575, ONR Award No. N00014-16-1- 2189, CCF-1763673, CCF-1755808 and the UC Office of President under grant No. LFR-18-548175. The views, opinions, and/or findings expressed are those of the author(s) and should not be interpreted as representing the official views or policies of the Department of Defense or the U.S. Government.

%% file: 9-appendices.tex
\section{}
\label{lemma1proof}
\textbf{McDiarmid's Inequality:} Let $X_1,\cdots,X_n$ be independent random variables taking values in $\mathcal{X}$. Further, let the function $f:\mathcal{X}^n \to \mathbb{R}$ be $L_i$-Lipschitz for all $i \in [n]$, that is
\begin{align}
| f(x_1,\cdots,x_i,\cdots,x_n) - f(x_1,\cdots,x'_i,\cdots,x_n)| \leq L_i,\nonumber
\end{align}
for any $x_1,\cdots,x_n,x'_i \in \mathcal{X}$ and $i\in [n]$. Then, for any $\epsilon>0$,
\begin{flalign}
\Pr\Big[f(X_1,\cdots,X_n)- & \Expc[f(X_1,\cdots,X_n)]\geq\epsilon\Big]  \nonumber\\
& \leq \exp\left(-\frac{2\epsilon^2}{ \sum_{i=1}^n L_i^2}\right),\nonumber
\end{flalign}
\begin{align}
\Pr\Big[\Expc[f(X_1,\cdots,X_n)]- & f(X_1,\cdots,X_n)\geq\epsilon\Big] \nonumber\\
& \leq \exp\left(-\frac{2\epsilon^2}{ \sum_{i=1}^n L_i^2}\right).\nonumber
\end{align}
For each $i$, the aggregate return at time $t$ satisfies $X_i(t)\in\{0,\ell_i\}$. Therefore, we can use McDiarmid's inequality as follows:
\begin{flalign}
\Pr\left[X(t)-\Expc[X(t)]\geq\epsilon\right]\leq \exp\left(-\frac{2\epsilon^2}{ \sum_{i=1}^n \ell_i^2}\right),\nonumber
\end{flalign}
\begin{equation}
\Pr\left[\Expc[X(t)]-X(t)\geq\epsilon\right]\leq \exp\left(-\frac{2\epsilon^2}{ \sum_{i=1}^n \ell_i^2}\right),\nonumber
\end{equation}
for any $\epsilon>0$. Now, we proceed to the proof of Lemma \ref{lemma1}.

 \begin{proof}[Proof of Lemma \ref{lemma1}] Let $t=\tau^*+\delta$ for some $\delta= \Theta\left(\frac{\log n}{\sqrt{n}}\right)$ and $\epsilon=\delta^2$. The claim is that $\Prob\big[X^*(t)\leq r-\epsilon\big]=  o\left(\frac{1}{ n}\right)$. From McDiarmid's inequality, we have
\begin{align}
\label{eq:10}
\Prob\big[X^*(t)\leq r-\epsilon \big] &\leq \exp \left(-\frac{2\big(\Expc[X^*(t)] - r+\epsilon \big)^2}{\sum_{i} \ell_{i}^*{^2}(t)}\right)\nonumber\\
&=\exp\left( -\frac{2\big(ts-r+\epsilon)^2}{\sum_{i} \ell_i^{*2}(t)} \right)\nonumber\\
&= \exp\left( -\frac{2\delta^2 s^2+2\delta^4+4\delta^3s}{ \big((\frac{r}{s})^2+\delta^2+2\delta \frac{r}{s} \big)\sum_i \lambda_i^2} \right)\nonumber\\
& \overset{(g)}{=} e^{-\Theta \left(n \delta^2\right)} = o\left(\frac{1}{ n}\right).\nonumber
\end{align}
In above, equality $(g)$ follows from the fact that $r=\Theta (n)$, $s=\Theta (n)$, $\lambda_i=\Theta(1)$, $\delta=\Theta\left(\frac{\log n}{\sqrt{n}}\right)$, and therefore $\sum_i \lambda_i^2=\Theta(n)$ and $s^2=\Theta(n^2)$. Moreover, if $t^*<\tau^*$, with a positive probability there are less than $r$ equations at the master node by time $t^*$ which is a contradiction. Therefore,  \begin{equation}
\tau^*\leq t^*\leq \tau^*+\delta.\nonumber
\end{equation}
\end{proof}

%% file: 10-bio.tex
\begin{IEEEbiographynophoto}{Amirhossein Reisizadeh}
received his B.S. degree form Sharif University of Technology, Tehran, Iran in 2014 and an M.S. degree from University of California, Los Angeles (UCLA) in 2016, both in Electrical Engineering. He is currently pursuing his Ph.D. in  Electrical and Computer Engineering at University of California, Santa Barbara (UCSB). He is interested in using information and coding-theoretic concepts to develop fast and efficient algorithms for large-scale machine learning, distributed computing and optimization.
\end{IEEEbiographynophoto}

\begin{IEEEbiographynophoto}{Saurav Prakash}
received his Bachelor of Technology degree in Electrical Engineering from the Indian Institute of Technology (IIT), Kanpur, India in 2016 and is currently pursuing his Ph.D. in Electrical and Computer Engineering from the University of Southern California (USC), Los Angeles. Saurav received the Annenberg Graduate Fellowship in 2016. He was one of the Viterbi-India fellows in summer 2015. His interests include information theory and data analytics with applications in large-scale machine learning and edge computing. 
\end{IEEEbiographynophoto}

\begin{IEEEbiographynophoto}{Ramtin Pedarsani}
is an Assistant Professor in ECE Department at the University of California, Santa Barbara. He received the B.Sc. degree in electrical engineering from the University of Tehran, Tehran, Iran, in 2009, the M.Sc. degree in communication systems from the Swiss Federal Institute of Technology (EPFL), Lausanne, Switzerland, in 2011, and his Ph.D. from the University of California, Berkeley, in 2015. His research interests include machine learning, information and coding theory, networks, and transportation systems. Ramtin is a recipient of the IEEE international conference on communications (ICC) best paper award in 2014.
\end{IEEEbiographynophoto}

\begin{IEEEbiographynophoto}{A. Salman Avestimehr}
is a Professor at the Electrical and Computer Engineering Department of University of Southern California. He received his Ph.D. in 2008 and M.S. degree in 2005 in Electrical Engineering and Computer Science, both from the University of California, Berkeley. Prior to that, he obtained his B.S. in Electrical Engineering from Sharif University of Technology in 2003. He was an Assistant Professor at the ECE school of Cornell University from 2009 to 2013. He was also a postdoctoral scholar at the Center for the Mathematics of Information (CMI) at Caltech in 2008. His research interests include information theory, coding theory, and large-scale distributed computing and machine learning.

Dr. Avestimehr has received a number of awards for his research, including an Information Theory Society and Communication Society Joint Paper Award, a Presidential Early Career Award for Scientists and Engineers (PECASE) from the White House, a Young Investigator Program (YIP) award from the U. S. Air Force Office of Scientific Research, a National Science Foundation CAREER award, the David J. Sakrison Memorial Prize, and several Best Paper Awards at Conferences. He has been an Associate Editor for IEEE Transactions on Information Theory. He is currently a general Co-Chair of the 2020 International Symposium on Information Theory (ISIT).

\end{IEEEbiographynophoto}